\RequirePackage{fix-cm}
\documentclass[smallextended]{amsart}       
\usepackage{graphicx}
\usepackage{mathptmx}      
\usepackage{graphicx,subfig}
\usepackage{amssymb,epsfig,amsmath,url,enumerate,booktabs,bm,multirow}
\usepackage[margin=1.5in]{geometry}
\newcommand{\E}{\mathbb{E}}
\newcommand{\Id}{\mathrm{d}}
\newcommand{\R}{\mathbb{R}}

\newcommand{\argmax}{\operatornamewithlimits{argmax}}
\newcommand{\argmin}{\operatornamewithlimits{argmin}}

\theoremstyle{plain}

\newtheorem{proposition}{Proposition}
\newtheorem{corollary}{Corollary}
\newtheorem{lemma}{Lemma}

\theoremstyle{definition}

\newtheorem{assumption}{Assumption}

\theoremstyle{remark}
\newtheorem{remark}{Remark}

\begin{document}

\title{Electricity price modeling and asset valuation: a multi-fuel structural approach}




\author{Ren\'{e} Carmona}
\address{Bendheim Center for Finance\\
Dept. ORFE, University of Princeton\\Princeton NJ 08544, USA}
\email{rcarmona@princeton.edu}

\author{Michael Coulon}
\thanks{Partially supported by NSF - DMS-0739195}
\address{ORFE\\University of Princeton\\Princeton NJ 08544, USA}
\email{mcoulon@princeton.edu}

\author{Daniel Schwarz}
\address{Oxford-Man Institute\\University of Oxford\\Oxford, UK}
\email{schwarz@maths.ox.ac.uk}


\maketitle

\begin{abstract}
We introduce a new and highly tractable structural model for spot and derivative prices in electricity markets. Using a stochastic model of the bid stack, we translate the demand for power and the prices of generating fuels into electricity spot prices. The stack structure allows for a range of generator efficiencies per fuel type and for the possibility of future changes in the merit order of the fuels. The derived spot price process captures important stylized facts of historical electricity prices, including both spikes and the complex dependence upon its underlying supply and demand drivers.  Furthermore, under mild and commonly used assumptions on the distributions of the input factors, we obtain closed-form formulae for electricity forward contracts and for spark and dark spread options. As merit order dynamics and fuel forward prices are embedded into the model, we capture a much richer and more realistic dependence structure than can be achieved by classical reduced-form models. We illustrate these advantages by comparing with Margrabe's formula and a simple cointegration model, and highlight important implications for the valuation of power plants.



\end{abstract}

\section{Introduction}\label{str:intro}
Since the onset of electricity market deregulation in the 1990s, the modeling of prices in these markets has become an important topic of research both in academia and industry. Energy companies own large portfolios of generation units and require sophisticated models of price dynamics in order to manage risk. Asset valuation is also of utmost importance in capital intensive industries and real option theory is typically used to associate the management of a plant to a string of spread options, spanning many years or decades. One of the main thrusts of this paper is to provide new and versatile tools for these valuations which efficiently capture the complex dependencies upon demand and production fuel prices.

In electricity price modeling, important challenges include prominent seasonalities and mean-reversion at various time scales, sudden spikes and the strong link between the prices of electricity and other energy commodities (see Figures 1b and 5 for sample daily historical spot and forward prices from the PJM market) -- features which mostly stem from the non-storability of electricity and the resulting matching of supply and demand at all times.  While every model should attempt to capture these properties as well as possible, at the same time there is a need for fast and efficient methods to value power plants and other derivatives on the spot price. To achieve the latter goal of efficiency much literature has ignored or oversimplified the former goal of modeling structural relationships. In this paper, we propose a model that realistically captures the dependency of power prices on their primary drivers; yet we obtain closed-form expressions for spot, forward and option prices.  

The existing literature on electricity price modeling can be approximately divided into three categories. At one end of the spectrum are so called full production cost models. These rely on knowledge of all generation units, their corresponding operational constraints and network transmission constraints.  Prices are then typically solved for by complex optimization routines (cf. \cite{aEydeland2003}). Although this type of model may provide market insights and forecasts in the short term, it is --- due to its complexity --- unsuited to handling uncertainty, and hence to derivative pricing or the valuation of physical assets.  Other related approaches which share this weakness include models of strategic bidding (cf. \cite{aHortacsu2007}) and other equilibrium approaches (cf. \cite{hBessembinder2002}).  At the other end of the spectrum are reduced form models. These are characterised by an exogenous specification of electricity prices, with either the forward curve (cf. \cite{lClewlow2000} and \cite{fBenth2007a}) or the spot price (cf. \cite{sDeng2003,nFrikha2009,fBenth2008}) representing the starting point for the model. Reduced form models typically either ignore fuel prices or introduce them as exogenous correlated processes; hence they are not successful at capturing the important afore mentioned dependence structure between fuels and electricity. Further, spikes are usually only obtained through the inclusion of jump processes or regime switches, which provide little insight into the causes that underly these sudden price swings. 

In between these two extremes is the structural approach to electricity price modeling, which stems from the seminal work of Barlow \cite{mBarlow2002}. We use the adjective structural to describe models, which --- to varying degrees of detail and complexity --- explicitly approximate the supply curve in electricity markets (commonly known as the bid stack due to the price-setting auction). The market price is then obtained under the equilibrium assumption that demand and supply must match. In Barlow's work the bid stack is simply a fixed parametric function, which is evaluated at a random demand level. Later works have refined the modeling of the bid curve and taken into account its dependency on the available capacity (cf. \cite{mBurger2004,aBoogert2008,aCartea2008a}), as well as fuel prices (cf. \cite{cPirrong2008,mCoulon2009a,rAid2009,rAid2011}) and the cost of carbon emissions (cf. \cite{sHowison2011,mCoulon2009}). The raison d'\^{e}tre of all structural models is very clear. If the bid curve is chosen appropriately, then observed stylized facts of historic data can be well matched. Moreover, because price formation is explained using fundamental variables and costs of production, these models offer insight into the causal relationships in the market; for example, prices in peak hours are most closely correlated with natural gas prices in markets with many gas `peaker' plants; similarly, price spikes are typically observed to coincide with states of very high demand or low capacity. As a direct consequence, this class of models also performs best at capturing the varied dependencies between electricity, fuel prices, demand and capacity.

The model we propose falls into the category of structural models. Our work breaks from the current status quo by providing closed-form formulae for the prices of a number of derivative products in a market driven by two underlying fuels and featuring a continuum of efficiencies (heat rates). In the considered multi-fuel setting, our model of the bid stack allows the merit order to be dynamic: each fuel can become the marginal fuel and hence set the market price of electricity. Alternatively, several fuels can set the price jointly. Despite this complexity, under only mild assumptions on the distribution, under the pricing measure, of the terminal value of the processes representing electricity demand and fuels, we obtain explicit formulae for spot prices, forwards and spread options, as needed for power plant valuation. Moreover, our formulae capture very clearly and conveniently the dependency of electricity derivatives upon the prices of forward contracts written on the fuels that are used in the production process. This allows the model to easily `see' additional information contained in the fuel forward curves, such as states of contango or backwardation --- another feature, which distinguishes it from other approaches.

The parametrization of the bid stack we propose combines an exponential dependency on demand, suggested several times in the literature (cf. \cite{pSkantze2000,aCartea2008,mLyle2009}), with the need for a heat rate function multiplicative in the fuel price, as stressed by Pirrong and Jermakyan \cite{cPirrong2008}.  Eydeland and Geman \cite{aEydeland1999} propose a similar structure for forward prices and note that Black-Scholes like derivative prices are available if the power price is log-normal. However, this requires the assumption of a single marginal fuel type and ignores capacity limits. Coulon and Howison \cite{mCoulon2009a} construct the stack by approximating the distribution of the clusters of bids from each technology, but their approach relies heavily on numerical methods when it comes to derivative pricing. In the work of Aid \emph{et al} \cite{rAid2009}, the authors simplify the stack construction by allowing only one heat rate (constant heat rate function) per fuel type, a significant oversimplification of spot price dynamics for mathematical convenience. Aid \emph{et al} \cite{rAid2011} extend this approach to improve spot price dynamics and capture spikes, but at the expense of a static merit order, ruling out, among other things, the possibility that coal and gas can change positions in the stack in the future. In both cases the results obtained by the authors only lead to semi-closed form formulae, which still have to be evaluated numerically. 

The importance of the features incorporated in our model is supported by prominent developments observed in recent data.  In particular, shale gas discoveries have led to a dramatic drop in US natural gas prices in recent years, from a high of over \$13 in 2008 to under \$3 in January 2012.  Such a large price swing has rapidly pushed natural gas generators down the merit order, and highlights the need to account for uncertainty in future merit order changes, particularly for longer term problems like plant valuation. In addition, studying hourly data from 2004 to 2010 on marginal fuels in the PJM market (published by Monitoring Analytics), we observe that the electricity price was fully set by a single technology (only one marginal fuel) in only 16.1 per cent of the hours.  For the year 2010 alone, the number drops to less than 5 per cent. Substantial overlap of bids from different fuels therefore exists, and changes in merit order occur gradually as prices move. We believe that our model of the bid stack adheres to many of the true features of the bid stack structure, which leads to a reliable reproduction of observed correlations and price dynamics, while retaining mathematical tractability.

\section{Structural approach to electricity pricing}\label{str:structural_approach}
In the following we work on a complete probability space $(\Omega, \mathcal{F},\mathbb{P})$. For a fixed time horizon $T \in \R_+$, we define the $(n+1)$-dimensional standard Wiener process $(W^0_t,\mathbf{W}_t)_{t\in[0,T]}$, where $\mathbf{W}:=(W^1,\ldots,W^n)$. Let $\mathcal{F}^0:=(\mathcal{F}^0_t)$ denote the filtration generated by $W^0$ and $\mathcal{F}^W:=(\mathcal{F}^W_t)$ the filtration generated by $\mathbf{W}$. Further, we define the market filtration $\mathcal{F}:=\mathcal{F}^0\vee\mathcal{F}^W$. All relationships between random variables are to be understood in the almost surely sense.

\subsection{Price Setting in Electricity Markets}\label{str:price_setting}
We consider a market in which individual firms generate electricity. All firms submit day-ahead bids to a central market administrator, whose task it is to allocate the production of electricity amongst them. Each firm's bids take the form of price-quantity pairs representing an amount of electricity the firm is willing to produce, and the price at which the firm is willing to sell it\footnote{Alternatively, firms may, in some markets, submit continuous bid curves, which map an amount of electricity to the price at which it is offered. For our purposes this distinction will however not be relevant.}. An important part is therefore played by the merit order, a rule by which cheaper production units are called upon before more expensive ones in the electricity generation process. This ultimately guarantees that electricity is supplied at the lowest possible price.\footnote{This description is of course a simplification of the market administrator's complicated unit commitment problem, typically solved by optimization in order to satisfy various operational constraints of generators, as well as transmission constraints.  Details vary from market to market and we do not address these issues here, as our goal is to approximate the price setting mechanism and capture the key relationships needed for derivative pricing.}

\begin{assumption}\label{as:merit_order}
 The market administrator arranges bids according to the merit order and hence in increasing order of costs of production.
\end{assumption}

We refer to the resulting map from the total supply of electricity and the factors that influence the bid levels to the price of the marginal unit as the \textit{market bid stack} and assume that it can be represented by a measurable function
\begin{equation*}\label{eq:bs_original}
b:[0,\bar{\xi}]\times\R^{n}\ni(\xi,\mathbf{s})\hookrightarrow b(\xi,\mathbf{s})\in\R,
\end{equation*}
which will be assumed to be strictly increasing in its first variable. Here, $\bar{\xi}\in\R_+$ represents the combined capacity of all generators in the market, henceforth the \textit{market capacity} (measured in MW), and $\mathbf{s}\in\R^n$ represents factors of production which drive firm bids (e.g. fuel prices).

\textit{Demand for electricity} is assumed to be price-inelastic and given exogenously by an $\mathcal{F}^0_t$-adapted process $(D_t)$ (measured in MW). As we shall see later, the prices of the factors of production used in the electricity generation process will be assumed to be $\mathcal{F}^W_t$-adapted. So under the objective historical measure $\mathbb{P}$, the demand is statistically independent of these prices. This is a reasonable assumption as power demand is typically driven predominantly by temperature, which fluctuates at a faster time scale and depends more on local or regional conditions than fuel prices. The market responds to this demand by supplying an amount $\xi_t \in [0,\bar{\xi}]$ of electricity. We assume that the market is in equilibrium with respect to the supply of and demand for electricity; i.e.
\begin{equation}\label{eq:demand_supply_equilibrium}
D_t = \xi_t, \quad \text{for } t \in [0,T]. 
\end{equation}
This implies that $D_t \in [0,\bar{\xi}]$ for $t \in [0,T]$ and $(\xi_t)$ is $\mathcal{F}^0_t$-adapted.

The \textit{market price of electricity} $(P_t)$ is now defined as the price at which the last unit that is needed to satisfy demand sells its electricity; i.e. using \eqref{eq:demand_supply_equilibrium},
\begin{equation}\label{eq:def_mp_el}
 P_t := b(D_t,\cdot), \quad \text{for } t\in [0,T].
\end{equation}

We emphasize the different roles played by the first variable (i.e. demand) and all subsequent variables (i.e. factors driving bid levels) of the bid stack function $b$. Due to the inelasticity assumption, the level of demand fully determines the quantity of electricity that is being generated; all subsequent variables merely impact the merit order arrangement of the bids.

\begin{remark}
The price setting mechanism described above applies directly to day-ahead spot prices set by uniform auctions, as in most exchanges today. However, we believe that in a competitive market with rational agents, the day-ahead auction price also serves as the key reference point for real-time and over-the-counter prices.
\end{remark}

\subsection{Mathematical Model of the Bid Stack}\label{str:bid_stack}
From the previous subsection, it is clear that the price of electricity in a structural model like the one we are proposing depends critically on the construction of the function $b$. Before we explain how this is done in the current setting, we make the following assumption about the formation of firms' bids.

\begin{assumption}\label{as:bs_cost_str}
Bids are driven by production costs. Furthermore,
\begin{enumerate}
  \item costs depend on fuel prices and firm-specific characteristics only;\label{as:bs_cost_str_1}
  \item firms' marginal costs are strictly increasing.\label{as:bs_cost_str_2}
\end{enumerate}
\end{assumption}

\begin{figure}[htbp]
  \centering
  \subfloat[Sample bid stacks.]{\label{fig:PJMstacks}\includegraphics[width=0.5\textwidth]{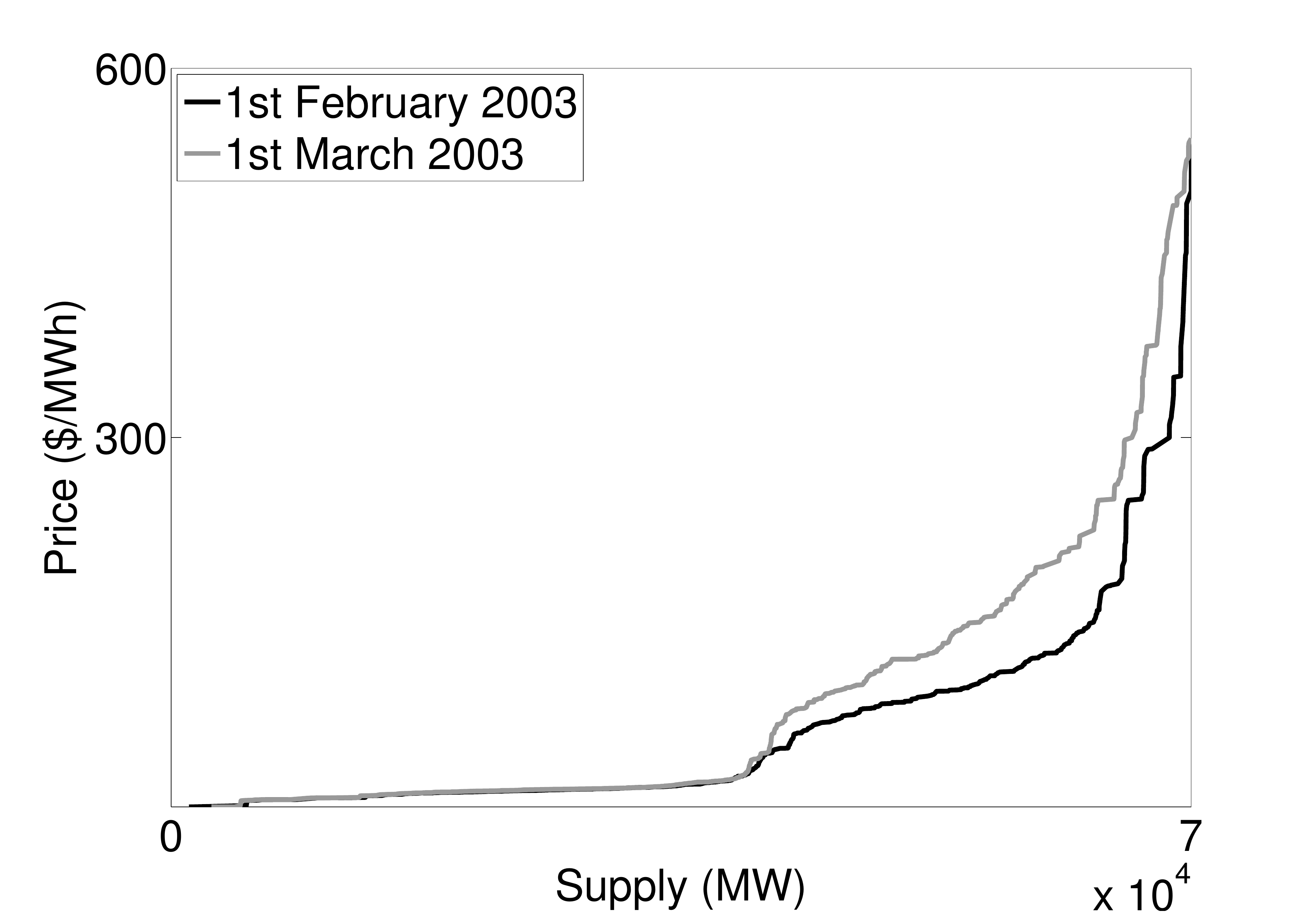}}
  \subfloat[Historical daily average prices.]{\label{fig:PJMprices}\includegraphics[width=0.5\textwidth]{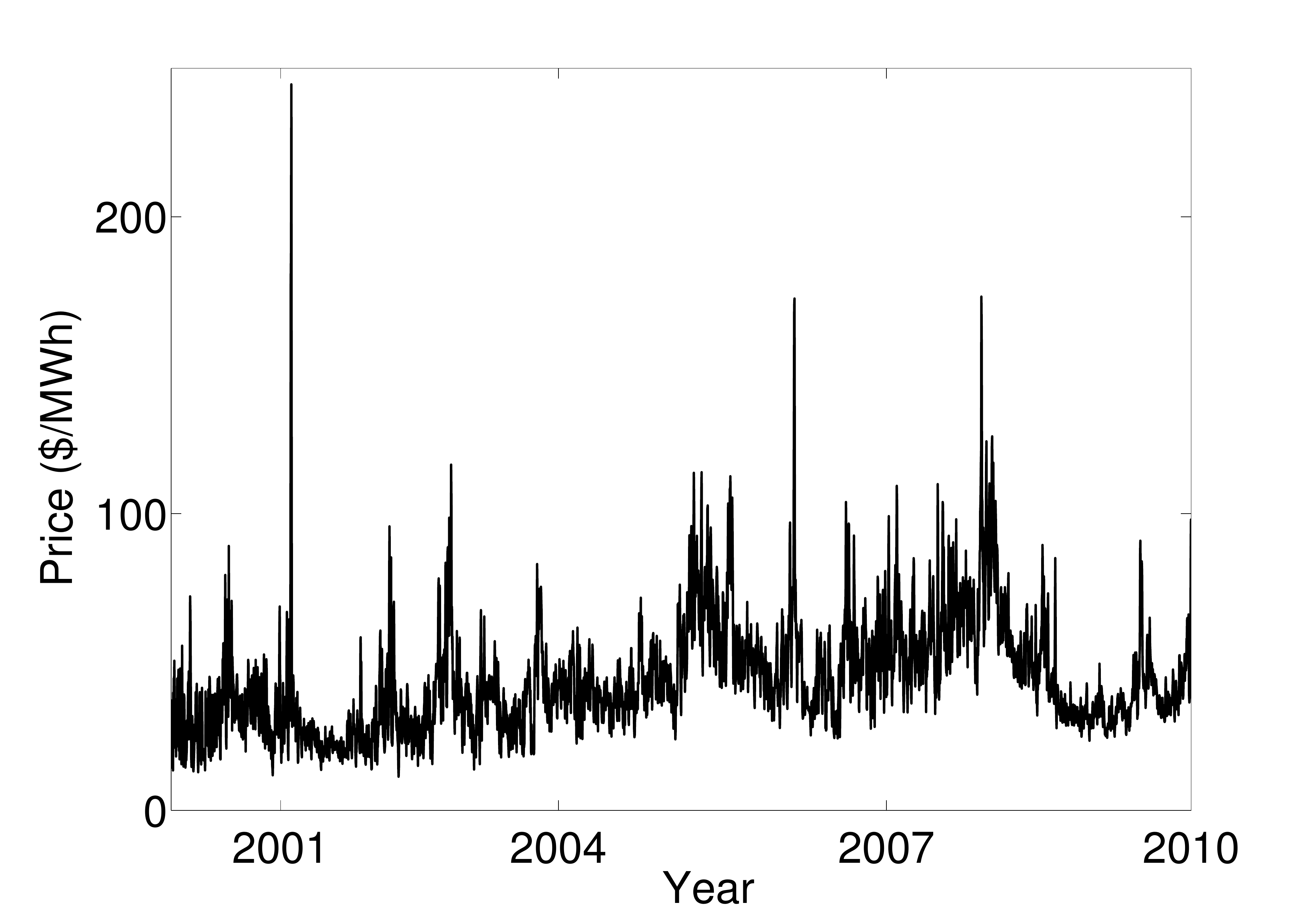}}
  \caption{Historical prices and bids from the PJM market in the North East US} 
 \label{fig:PJMdata}
\end{figure}

To back up Assumption \ref{as:bs_cost_str}, we briefly consider it in light of historic bid and spot price data. Figure \ref{fig:PJMstacks} plots the bid stack from the PJM market region in the US on the first day of two consecutive months. Firstly, rapidly increasing marginal costs lead to the steep slope of the stack near the market capacity (70,000 MW); this feature translates directly into spikes in spot prices (see Figure \ref{fig:PJMprices}) when demand is high. Between the two sample dates in February and March 2003, the prices of gas increased rapidly. In the PJM market gas fired plants have historically featured mainly in the second half of the bid stack. Therefore, it is precisely the gas price related increase in production costs, which explains the increase in bid levels observed beyond about 40,000MW in Figure \ref{fig:PJMstacks} (see also \textsection \ref{str:correlations} for a discussion of recent merit order changes).

Production costs are typically linked to a particular fuel price (e.g. coal, natural gas, lignite, oil, etc.). Furthermore, within each fuel class, the cost of production may vary significantly, for example as old generators may have a higher heat rate (lower efficiency) than new units.  It is not our aim to provide a mathematical model that explains how to aggregate individual bids or capture strategic bidding. Instead, we group together generators of the same fuel type and assume the resulting bid curve to be exogenously given and to satisfy Assumption \ref{as:bs_cost_str}. From this set of bid curves, the merit order rule then determines the construction of the market bid stack.

Let $I=\{1,\ldots,n\}$ denote the index set of all the fuels used in the market to generate electricity. We assume that their prices are the only factors influencing the bid stacks.  With each $i\in I$ we associate an $\mathcal{F}^W_t$-adapted \textit{fuel price process} $(S^i_t)$ and we define the \textit{fuel bid curve} for fuel $i$ to be a measurable function
\begin{equation*}
 b_i:[0,\bar\xi_i]\times \R\ni (\xi,s)\hookrightarrow b_i(\xi,s)\in\R,
\end{equation*}
where the argument $\xi$ represents the amount of electricity supplied by generators utilizing fuel type $i$, $s$ a possible value of the price $S^i_t$, and $\bar{\xi}^i\in\R_+$ the aggregate capacity of all the generators utilizing fuel type $i$. We assume that $b_i$ is strictly increasing in its first argument, as required by Assumption \ref{as:bs_cost_str}. Further, also for $i\in I$, let the $\mathcal{F}_t$-adapted process $(\xi^i_t)$ represent the amount of electricity supplied by generators utilizing fuel type $i$. It follows that
\begin{equation*}
\sum_{i\in I} \bar{\xi}^i = \bar{\xi},\quad \text{and} \quad D_t = \sum_{i\in I}\xi^i_t, \quad \text{for } t\in[0,T].
\end{equation*}

In order to simplify the notation below, for $i\in I$, and for each $s\in\R$ we denote by $b_i(\,\cdot\,,s)^{-1}$
the generalized (right continuous) inverse of the function $\xi\hookrightarrow b_i(\xi,s)$. Thus
$$
b_i(\,\cdot\,,s)^{-1}(p):=\bar{\xi}^i\wedge\inf\left\{\xi\in(0,\bar{\xi}^i]:\,b_i(\xi,s)>p\right\}, \quad \text{for } (p,s)\in \R \times \R,
$$
where we use the standard convention $\inf \emptyset=+\infty$. Using the notation
\begin{equation*}
\underline{b}_i(s):=b_i(0,s) \quad \text{and} \quad \bar{b}_i(s):= b_i(\bar{\xi}^i,s)
\end{equation*}
and writing $\hat b_i^{-1}(p,s)=b_i(\,\cdot\,,s)^{-1}(p)$ to ease the notation, we see that $\hat b_i^{-1}(p,s)=0$ if $p \in (-\infty,\underline{b}_i(s))$, $\hat b_i^{-1}(p,s)=\bar{\xi}^i$ if $p \in [\bar{b}_i(s),\infty)$, and $\hat b_i^{-1}(p,s)\in [0,\bar{\xi}^i]$ if $p \in [\underline{b}_i(s),\bar{b}_i(s))$. For fuel $i\in I$ at price $S^i_t=s$ and for electricity prices below $\underline{b}_i(s)$ no capacity from the $i$th technology will be available. Similarly, once all resources from a technology are exhausted, increases in the electricity price will not lead to further production units being brought online. So defined, the inverse function $\hat{b}_i^{-1}$ maps a given price of electricity and the price of fuel $i$ to the amount of electricity supplied by generators relying on this fuel type.
\begin{proposition}\label{prop:el_price}
 For a given vector $(D_t,\mathbf{S}_t)$, where $\mathbf{S}_t:=(S^1_t,\ldots,S^n_t)$, the market price of electricity $(P_t)$ is determined by
 \begin{equation}\label{eq:el_price_implicit}
  P_t = \min_{i\in I} \underline{b}_i\left(S^i_t\right) \vee \sup \left\{p  \in \R : \sum_{i\in I} \hat{b}_i^{-1}\left(p ,S^i_t\right) < D_t \right\}, \quad \text{for } t \in [0,T].
 \end{equation}
\end{proposition}
\begin{proof}
By the definition of $\hat{b}_i^{-1}$, the function $\tilde{b}^{-1}$, defined by 
\begin{equation*}
\tilde{b}^{-1}(p,s^1,\cdots,s^n):=\sum_{i\in I}\hat{b}_i^{-1}(p,s^i),
\end{equation*}
is, when the prices of all the fuels are fixed, a non-decreasing map taking the electricity price to the corresponding amount of electricity generated by the market. Similarly to the case of one fixed fuel price,  for each fixed set of fuel prices, say $\mathbf{s}:=(s^1,\cdots,s^n)$, we define
the  bid stack function $\xi\hookrightarrow b(\xi,\mathbf{s})$ as the generalized (left continuous) inverse of the function $\xi\hookrightarrow \tilde{b}^{-1}(p,s^1,\cdots,s^n)$ defined above, namely
\begin{equation*}\label{eq:inv_bs_preimage}
b(\xi,\mathbf{s}):= \min_{i\in I}\underline{b}_i\left(s^i\right) \vee \sup \left\{p  \in \R : \sum_{i\in I} \hat{b}_i^{-1}\left(p ,s^i\right) < \xi \right\}, \quad \text{for } (\xi,\mathbf{s})\in [0,\bar{\xi}]\times \R^n,
\end{equation*}
where we use the convention $\sup \emptyset=-\infty$.

The desired result follows from the definition of the market price of electricity in \eqref{eq:def_mp_el}.
\end{proof}

\subsection{Defining a Pricing Measure in the Structural Setting}\label{str:measure}
The results presented in this paper do not depend on a particular model for the evolution of the demand for electricity and the prices of fuels. In particular, the concrete bid stack model for the electricity spot price introduced in \textsection \ref{str:exp_stack} is simply a deterministic function of the exogenously given factors under the real world measure $\mathbb{P}$. However, for the pricing of derivatives in \textsection \ref{str:forwards} and \textsection \ref{str:spreads} we need to define a pricing measure $\mathbb{Q}$ and the distribution of the random factors at maturity under this new measure will be important for the results that we obtain later.

For an $\mathcal{F}_t$-adapted process $\bm{\theta}_t$, where $\bm{\theta}_t:=(\theta^0_t,\theta^1_t,\ldots,\theta^n_t)$, a measure $\mathcal{Q}\sim\mathbb{P}$ is characterized by the Radon-Nikodym derivative
\begin{equation}\label{eq:measure_change}
 \frac{\Id \mathcal{Q}}{\Id \mathbb{P}}:=\exp\left(-\int_0^T\bm{\theta}_u\cdot\ \Id \mathbf{W}_u - \frac{1}{2}\int_0^T\left|\bm{\theta}_u\right|^2\ \Id u\right),
\end{equation}
where we assume that $(\bm{\theta}_t)$ satisfies the so-called Novikov condition
\begin{equation*}
\E\left[ \exp\left(\frac{1}{2}\int_0^T\left|\bm{\theta}_u\right|^2\ \Id u\right)\right]< \infty.
\end{equation*}
We identify $(\theta^0_t)$ and $(\theta^i_t)$ with the market prices of risk for demand and for fuel $i$ respectively.  We choose to avoid the difficulties of estimating the market price of risk (see for example \cite{aEydeland2003} for several possibilities) and instead make the following assumption.
\begin{assumption}\label{as:pricing_measure}
 The market chooses a pricing measure $\mathbb{Q}\sim\mathbb{P}$, such that
 \begin{equation*}
 \mathbb{Q}\in\left\{ \mathcal{Q} \sim \mathbb{P} : \text{all discounted prices of traded assets are local $\mathcal{Q}$-martingales}\right\}.
\end{equation*}
\end{assumption}

Note that we make no assumption regarding market completeness here. Because of the non-storability condition, certainly electricity cannot be considered a traded asset. Further, there are different approaches to modeling fuel prices; they may be treated as traded assets (hence local martingales under $\mathbb{Q}$) or --- more realistically --- assumed to exhibit mean reversion under the pricing measure. Either way, demand is a fundamental factor and the noise $(W^0_t)$ associated with it means that the joined market of fuels and electricity is bound to be incomplete. Note however, that all derivative products that we price later in the paper (forward contracts and spread options) are clearly traded assets and covered by Assumption \ref{as:pricing_measure}.
 
\section{Exponential Bid Stack Model}\label{str:exp_stack}
Equation \eqref{eq:el_price_implicit} in general cannot be solved explicitly. The reason for this is that any explicit solution essentially requires the inversion of the sum of inverses of individual fuel bid curves. 

We now propose a specific form for the individual fuel bid curves, which allows us to obtain a closed form solution for the market bid stack $b$. Here and throughout the rest of the paper, for $i\in I$, we define $b_i$ to be explicitly given by
\begin{equation}\label{eq:exp_bs}
 b_i(\xi,s):=s\exp(k_i + m_i \xi), \quad \text{for } (\xi,s)\in [0,\bar{\xi}^i]\times \R_+,
\end{equation}
where $k_i$ and $m_i$ are constants and $m_i$ is strictly positive. Note that $b_i$ clearly satisfies \eqref{as:bs_cost_str_1} and since it is strictly increasing on its domain of definition it also satisfies \eqref{as:bs_cost_str_2}.

\subsection{The Case of $n$ Fuels}\label{str:exp_n_fuel}
For observed $(D_t,\mathbf{S}_t)$, let us define the sets $M , C \subseteq I$ by
\begin{align*}
 M &:=\left\{ i \in I: \text{generators using fuel $i$ are partially used} \right\}\\
\text{and}\quad C &:=\left\{ i \in I: \text{the entire capacity $\bar{\xi}^i$ of generators using fuel $i$ is used}\right\}.
\end{align*}
A possible procedure for establishing the members of $M $ and $C $ is to order all the values of $\underline{b}_i$ and $\bar{b}_i$ and determine the corresponding cumulative amounts of electricity that are supplied at these prices. Then find where demand lies in this ordering.

With the above definition of $M $ and $C $ we arrive at the following corollary to Proposition \ref{prop:el_price}.

\begin{corollary}\label{cor:mp_el_exp_bs}
 For $b_i$ of exponential form, as defined in \eqref{eq:exp_bs}, the market price of electricity is given explicitly by the left continuous version of
\begin{equation}\label{eq:el_price_exp_n}
P_t=\left(\prod_{i\in M } \left(S^i_t\right)^{\alpha_i}\right) \exp\left\{\beta +\gamma\left(D_t-\sum_{i\in C } \bar{\xi}^i\right)\right\}, \quad \text{for } t\in[0,T],
\end{equation}
where
\begin{eqnarray*}
&&\alpha_i:= \frac{1}{\zeta}\left(\prod_{j\in M ,j\ne i} m_j\right),\quad
\beta:= \frac{1}{\zeta}\left(\sum_{l\in M } k_l\prod_{j\in M ,j\ne l}m_j\right),\\
&&\gamma:= \frac{1}{\zeta}\left(\prod_{j\in M }m_j\right)\quad\text{and}\quad\zeta:=\sum_{l\in M } \prod_{j\in M ,j\ne l} m_j.
\end{eqnarray*}
\end{corollary}

\begin{proof}
At any time $t\in [0,T]$ the electricity price depends on the composition of the sets $M $ and $C $; i.e. the current set of marginal and fully utilized fuel types.
 
For $i \in M $, $\hat{b}_i^{-1} = b_i^{-1}$, for $i \in C $, $\hat{b}_i^{-1} = \bar{\xi}^i$ and for $i \in I\setminus \{ M  \cup C  \}$, $\hat{b}_i^{-1} = 0$. Therefore, inside the supremum in \eqref{eq:el_price_implicit}, we replace $I$ with $M$ and take $\sum_{i\in C }\bar{\xi}^i$ to the right hand side. By Proposition \ref{prop:el_price} the electricity price is given by the left continuous inverse of the function $\sum_{i\in M }\hat{b}_i^{-1}$, which in the exponential case under consideration, simplifies to a single $\log$ function and yields \eqref{eq:el_price_exp_n}.
\end{proof}

\begin{figure}[htbp]
  \centering
  \subfloat[Fuel bid curves $b_i$.]{\label{fig:sec2_fuel_bs}\includegraphics[width=0.5\textwidth]{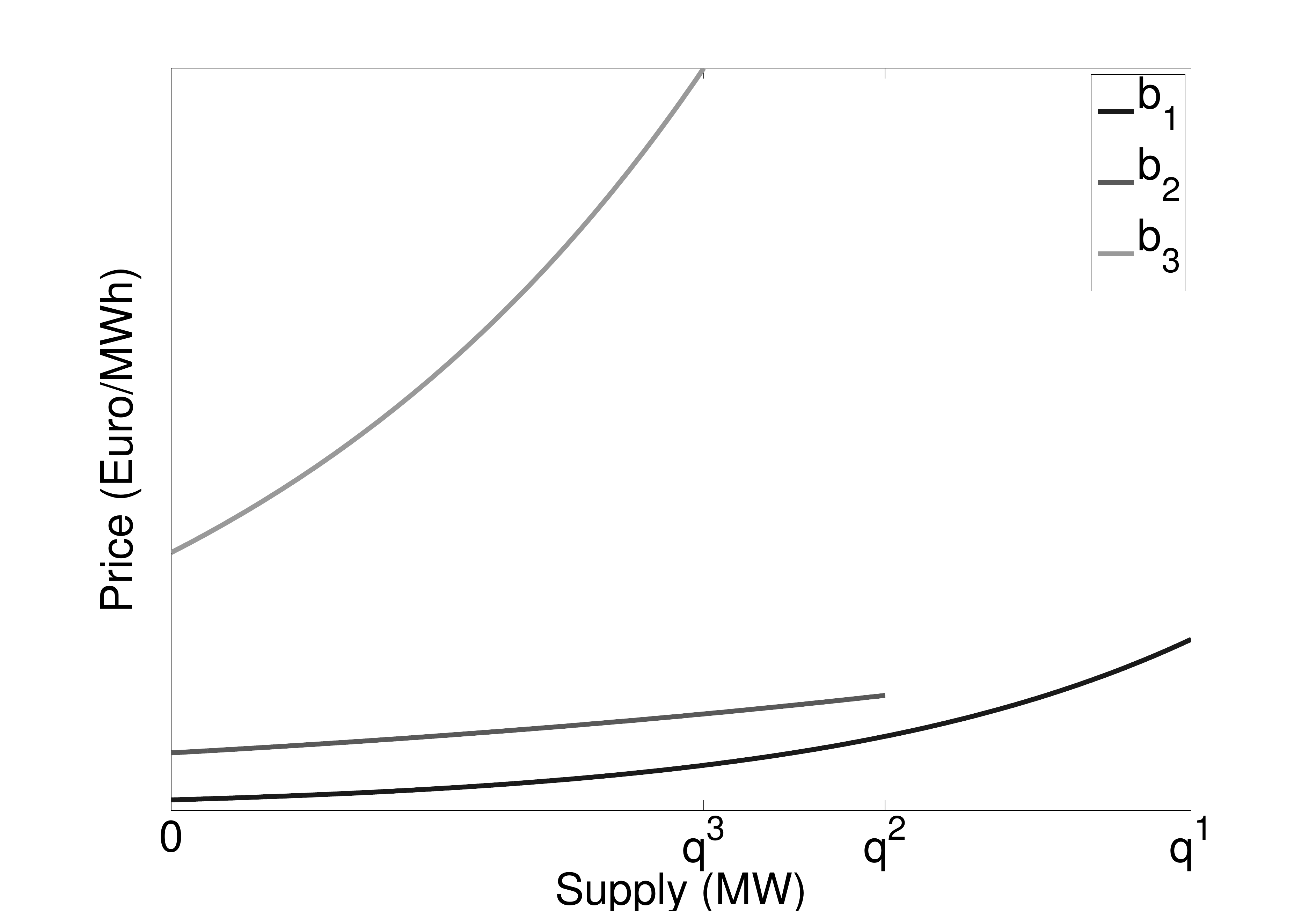}}
  \subfloat[Market bid stack $b$.]{\label{fig:sec2_market_bs}\includegraphics[width=0.5\textwidth]{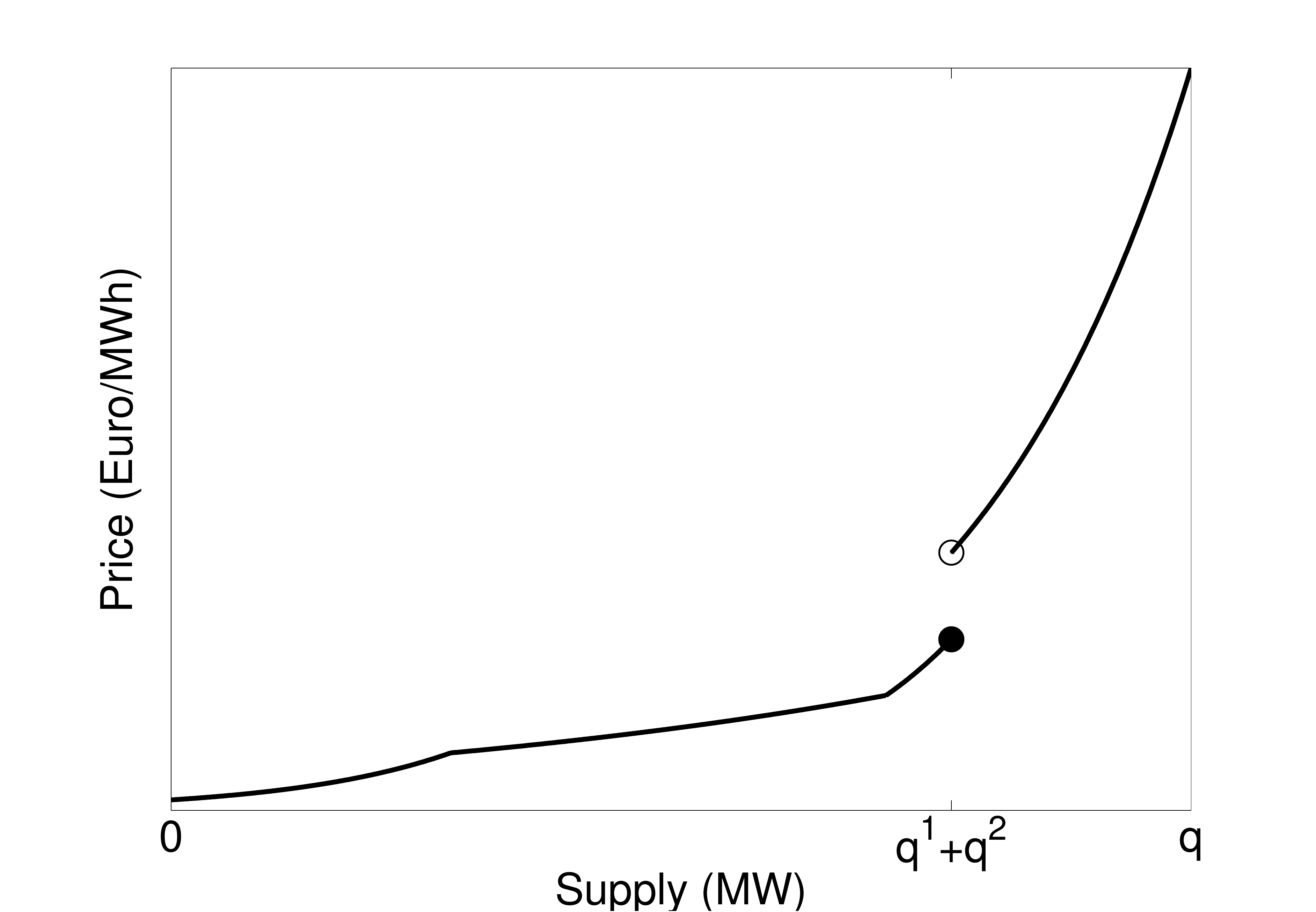}}
  \caption{Example of fuel bid curves and resulting market bid stack for $I:=\{1,2,3\},\ q:=\bar{\xi}$}
 \label{fig:sec2_bs}
\end{figure}

It is clear from equation \eqref{eq:el_price_exp_n} that the number of possible expressions for the electricity price is fully determined by the different configurations the sets $M $ and $C $ can take. In fact, fluctuations in demand and fuel prices can lead to 
\begin{equation}\label{NFuelPermutations}
 \sum_{i=1}^n \left(\begin{array}{c} n \\ i \end{array}\right) \left[ \sum_{j=0}^{n-i} \left(\begin{array}{c} n-i \\ j \end{array}\right) \right]
\end{equation}
distinct cases for \eqref{eq:el_price_exp_n}. Nonetheless, the market bid stack is always a piece-wise exponential function of demand (see Figure \ref{fig:sec2_bs}) with constantly evolving shape as fuel prices move. This captures in a very tractable way the influence of fuel prices on the merit order and resulting power price dynamics.

\subsection{The Case of Two Fuels}\label{str:exp_two_fuel}
For the remainder of the paper, we restrict our attention to the case of a two-fuel market, consisting of coal and natural gas generators. Our results can in principle be extended to the general case of $n>2$ fuels. However, the level of complexity of the formulas increases rapidly, as evidenced by the number of possible expressions given in \eqref{NFuelPermutations}. We also choose to omit the analysis of the one fuel case, which leads to far simpler expressions throughout, but cannot lead to merit order changes. From now on,  we set $I:=\{c,g\}$ and carry over all notation introduced in \textsection \ref{str:structural_approach} and \textsection \ref{str:exp_stack} so far.

From \eqref{NFuelPermutations} we know that there are five possible expressions for the electricity spot price. We list them in Table \ref{tbl:el_price_2fuel}.  Note that fixing $D_t$ reduces this list to some subset of three, each of which --- depending on the state of $\mathbf{S}_t$ --- can set the electricity price. (A similar reduction to three expressions occurs by fixing $\mathbf{S}_t$.) We exploit this property to write formula \eqref{eq:el_price_exp_n} in a form more amenable to calculations, identifying all cases explicitly.
\begin{table}[tbp]
 \begin{center}
  \begin{tabular}{cccc}
    \toprule
    \multirow{2}{*}{$P_t$, for $t\in[0,T]$} & \multirow{2}{*}{Criterion} & \multicolumn{2}{c}{Composition of} \\
     && $M $ & $C $\\
    \midrule
    $S^c_t\exp\left(k_c+m_cD_t\right)$ 				     & $b_c\left(D_t,S^c_t\right)\leq \underline{b}_g\left(S^g_t\right)$ & $\{c\}$ & $\{\varnothing\}$\\
    $S^g_t\exp\left(k_g+m_gD_t\right)$ 				 & $b_g\left(D_t,S^g_t\right)\leq \underline{b}_c\left(S^c_t\right)$ & $\{g\}$ & $\{\varnothing\}$ \\
    $S^c_t\exp\left(k_c+m_c\left(D_t-\bar{\xi}^g\right)\right)$ & $b_c\left(D_t-\bar{\xi}^g,S^c_t\right) > \bar{b}_g\left(S^g_t\right)$ & $\{c\}$ & $\{g\}$ \\
    $S^g_t\exp\left(k_g+m_g\left(D_t-\bar{\xi}^c\right)\right)$ & $b_g\left(D_t-\bar{\xi}^c,S^g_t\right) > \bar{b}_c\left(S^c_t\right)$ & $\{g\}$ & $\{c\}$ \\
    $\left(S^c_t\right)^{\alpha_c}\left(S^g_t\right)^{\alpha_g}\exp\left(\beta + \gamma D_t\right)$ & otherwise & $\{c,g\}$ & $\{\varnothing\}$ \\
   \bottomrule
  \end{tabular}
 \end{center}
 \caption{Distinct cases for the electricity price \eqref{eq:el_price_exp_n} in the two fuel case}  \label{tbl:el_price_2fuel}
\end{table}
To simplify the presentation in the text below, we define
\begin{equation*}
 b_{cg}\left(\xi,\mathbf{s}\right):=\left(s^c\right)^{\alpha_c}\left(s^g\right)^{\alpha_g}\exp\left(\beta + \gamma \xi \right), \quad \text{for } \left(\xi,\mathbf{s}\right)\in [0,\bar{\xi}]\times\R_+^2,
\end{equation*}
where $\alpha_c$, $\alpha_g$, $\beta$ and $\gamma$ are defined in Corollary \ref{cor:mp_el_exp_bs} and simplify for two fuels to
\begin{equation*}
\alpha_c=\frac{m_g}{m_c+m_g},\quad\alpha_g=1-\alpha_c=\frac{m_c}{m_c+m_g},\quad\beta=\frac{k_c m_g+k_gm_c}{m_c+m_g},\quad\gamma=\frac{m_c m_g}{m_c+m_g}.
\end{equation*}
 Further, we set $i_{-}:=\argmin \left\{\bar{\xi}^c,\bar{\xi}^g\right\}$ and $ i_{+}:=\argmax\left\{\bar{\xi}^c,\bar{\xi}^g\right\}$.

\begin{corollary}\label{cor:el_price_two_fuels}
With $I:=\{c,g\}$, for $t\in[0,T]$, the electricity spot price is given by
\begin{equation*}
 P_t = b_{\text{low}}\left(D_t,\mathbf{S}_t\right)\mathbb{I}_{\left[0,\bar{\xi}^{i_-}\right]}(D_t) + b_{\text{mid}}\left(D_t,\mathbf{S}_t\right)\mathbb{I}_{\left(\bar{\xi}^{i_-},\bar{\xi}^{i_+}\right]}(D_t) + b_{\text{high}}\left(D_t,\mathbf{S}_t\right)\mathbb{I}_{\left(\bar{\xi}^{i_+},\bar{\xi}\right]}(D_t),
\end{equation*}
 where, for $(\xi,\mathbf{s})\in [0,\bar{\xi}]\times\R_+^2$,
\begin{multline*}
b_{\text{low}}\left(\xi,\mathbf{s}\right) := b_c\left(\xi,s^c\right)\mathbb{I}_{\left\{b_c\left(\xi,s^c\right)< \underline{b}_g\left(s^g\right)\right\}} +  b_g\left(\xi,s^g\right)\mathbb{I}_{\left\{b_g\left(\xi,s^g\right)<\underline{b}_c\left(s^c\right)\right\}}\\
+ b_{cg}\left(\xi,\mathbf{s}\right)\mathbb{I}_{\left\{b_c\left(\xi,s^c\right) \geq \underline{b}_g\left(s^g\right), b_g\left(\xi,s^g\right) \geq \underline{b}_c\left(s^c\right)\right\}},
\end{multline*}
\begin{multline*}
b_{\text{mid}}\left(\xi,\mathbf{s}\right) := b_{i_+}\left(\xi,s^{i_+}\right)\mathbb{I}_{\left\{b_{i_+}\left(\xi,s^{i_+}\right)< \underline{b}_{i_-}\left(s^{i_-}\right)\right\}} + b_{i_+}\left(\xi-\bar{\xi}^{i_-},s^{i_+}\right)\mathbb{I}_{\left\{b_{i_+}\left(\xi-\bar{\xi}^{i_+},s^{i_+}\right)>\bar{b}_{i_-}\left(s^{i_-}\right)\right\}}\\
+b_{cg}\left(\xi,\mathbf{s}\right)\mathbb{I}_{\left\{b_{i_+}\left(\xi,s^{i_+}\right) \geq \underline{b}_{i_-}\left(s^{i_-}\right), b_{i_+}\left(\xi-\bar{\xi}^{i_-},s^{i_+}\right) \leq \bar{b}_{i_-}\left(s^{i_-}\right)\right\}},
\end{multline*}
\begin{multline*}
b_{\text{high}}\left(\xi,\mathbf{s}\right) := b_c\left(\xi-\bar{\xi}^g,s^c\right)\mathbb{I}_{\left\{b_c\left(\xi-\bar{\xi}^g,s^c\right)>\bar{b}_g\left(s^g\right)\right\}}+b_g\left(\xi-\bar{\xi}^c,s^g\right)\mathbb{I}_{\left\{b_g\left(\xi-\bar{\xi}^c,s^g\right)>\bar{b}_c\left(s^c\right)\right\}}\\
+b_{cg}\left(\xi,\mathbf{s}\right)\mathbb{I}_{\left\{b_c\left(\xi-\bar{\xi}^g,s^c\right) \leq \bar{b}_g\left(s^g\right),b_g\left(\xi-\bar{\xi}^c,s^g\right) \leq \bar{b}_c\left(s^c\right)\right\}}.
\end{multline*}
\end{corollary}

\begin{proof}
 The expressions for $b_{\text{low}}$, $b_{\text{mid}}$, $b_{\text{high}}$ are obtained from \eqref{eq:el_price_exp_n} by fixing $D_t$ in the intervals $(0,\bar{\xi}^{i_-}]$, $(\bar{\xi}^{i_-},\bar{\xi}^{i_+}]$, $(\bar{\xi}^{i_+},\bar{\xi}]$ respectively and considering the different scenarios for $M $ and $C $.
\end{proof}

\begin{figure}[htbp]
  \centering
  \subfloat[Surface plots of $b_{\text{low}}$, $b_{\text{mid}}$ and $b_{\text{high}}$]{\label{fig:sec2_bHigh}\includegraphics[width=0.5\textwidth]{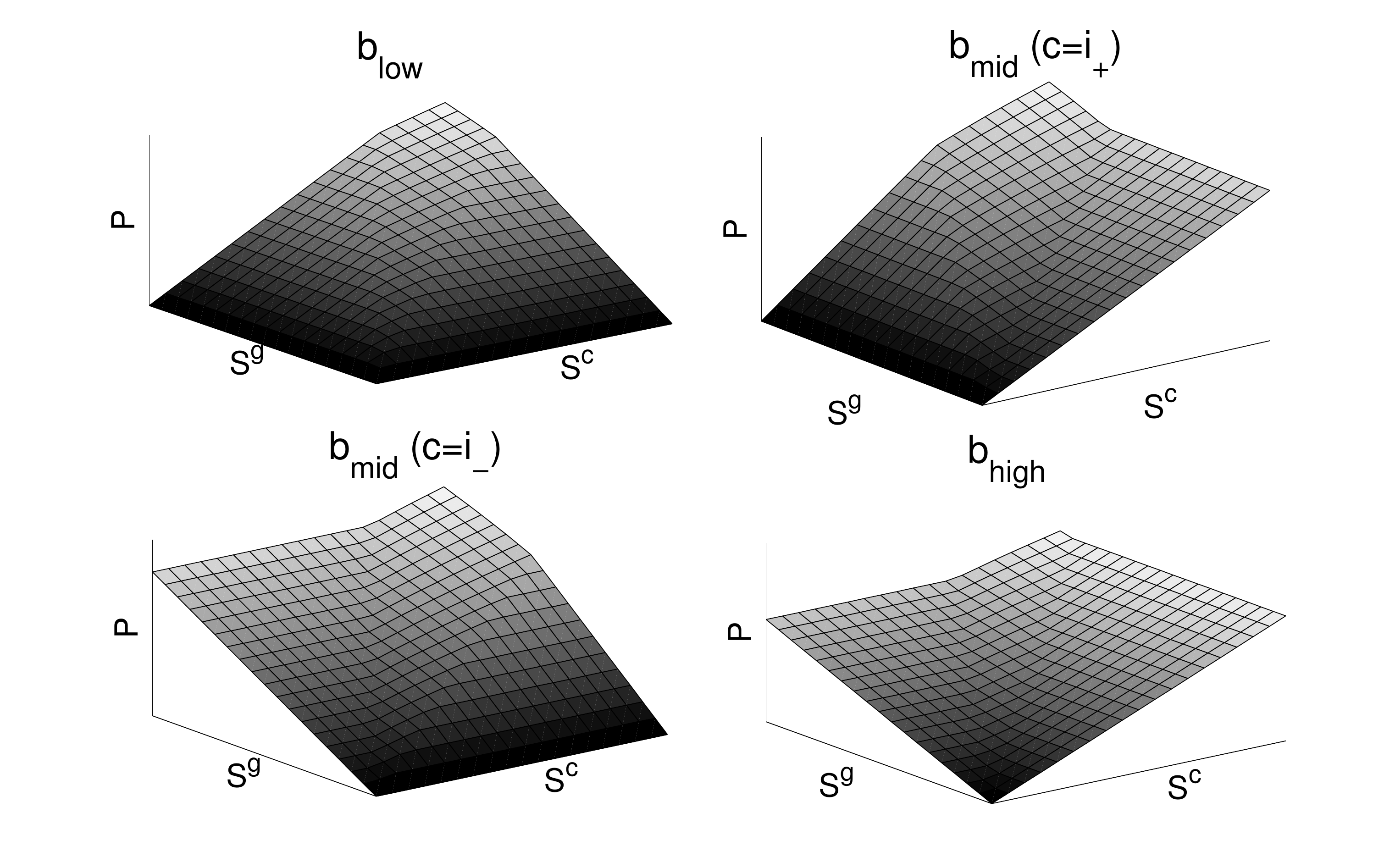}}
  \subfloat[$b_{\text{low}}$, $b_{\text{mid}}$, $b_{\text{high}}$ for fixed $S^c$ ($\bar{\xi}^c>\bar{\xi}^g$).]{\label{regionsplot}\includegraphics[width=0.5\textwidth]{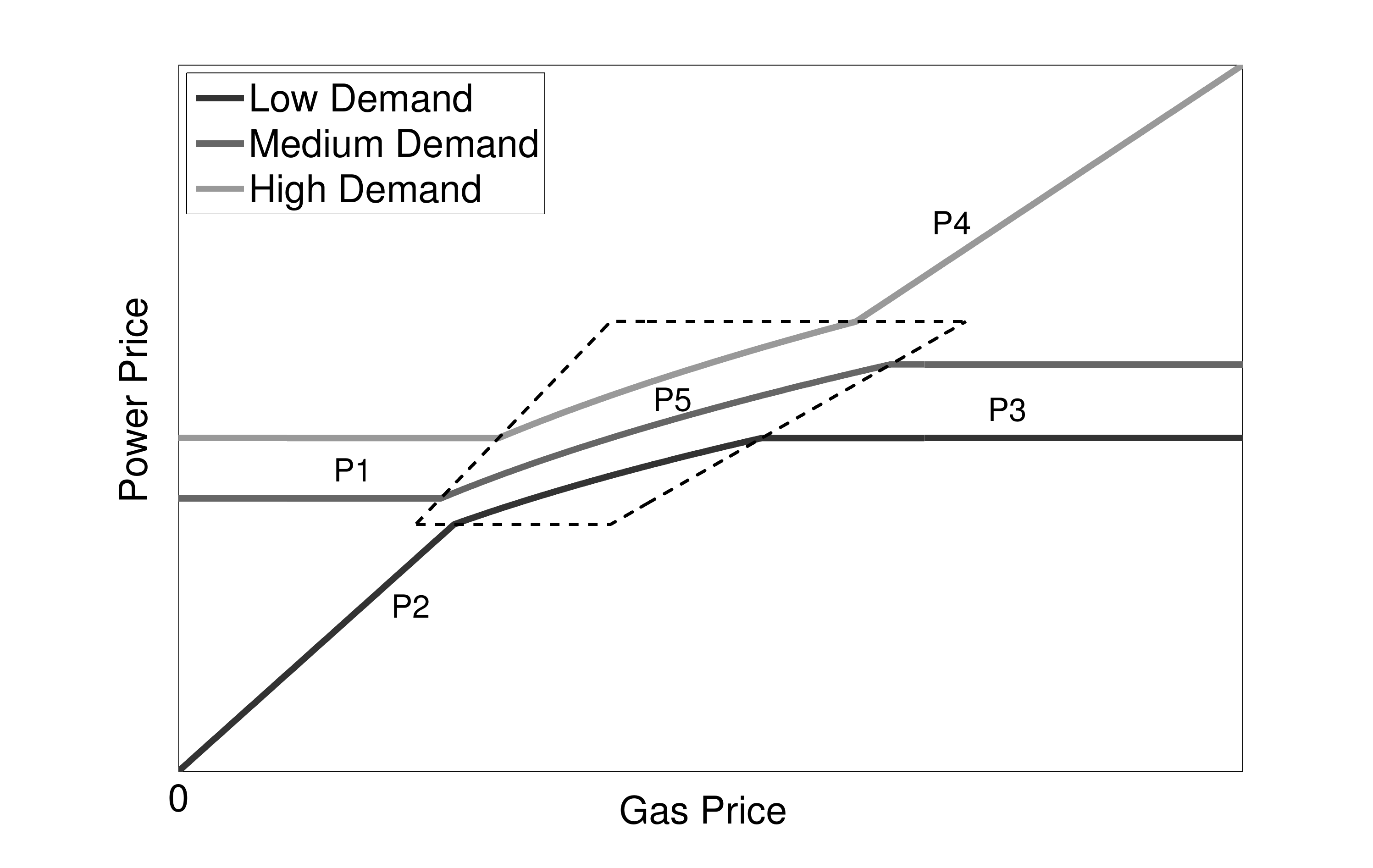}}
\caption{Illustration of the dependence of power spot price on fuel prices and demand}
 \label{fig:sec2_bLowMidHigh}
\end{figure}

Figure \ref{fig:sec2_bLowMidHigh} illustrates the dependence of the power price on the input factors, namely the prices of coal and gas and the demand for electricity. For the purpose of these plots we assumed throughout that $\bar{\xi}^c>\bar{\xi}^g$. The surfaces in Figure \ref{fig:sec2_bHigh} represent the functions $b_{\text{low}}$, $b_{\text{mid}}$ and $b_{\text{high}}$ for demand fixed in the corresponding intervals. In Figure \ref{regionsplot} we aggregate the information contained in Figure \ref{fig:sec2_bHigh} for a fixed coal price; i.e. we compare the dependency of the electricity price on the gas price for each of the three relevant demand levels. In all three cases electricity is non-decreasing in fuel price and is constant against $S^i_t$ if fuel $i$ is not at the margin (i.e., $i\notin M $).  Furthermore, in each case $P_t$ is linear in the sole marginal fuel for sufficiently low or high gas price, and non-linear in both fuels in the quadrilateral in the middle. This characterises the region of coal-gas overlap, where both technologies jointly set the price.  Finally, note that labels P1 to P5 indicate regions corresponding to rows one to five of Table \ref{tbl:el_price_2fuel}.

\subsection{Extension to Capture Spikes and Negative Prices}\label{str:spikes}
In this section we suggest a simple extension of the bid stack model in order to more accurately capture the spot price density in markets that are prone to dramatic price spikes during peak hours or sudden negative prices at off-peak times. Importantly, this modification does not impact the availability of closed-form solutions for forwards or spread options, which we introduce in \textsection \ref{str:forwards} and \textsection \ref{str:spreads}.

Let $(X_t)$ be a stochastic process adpated to the filtration $\mathcal{F}_t^0$, as is $D_t$. Further, we assume that the relationship between the two processes satisfies
\begin{equation*}
 \{X_t \leq 0\} = \{\omega\in\Omega : D_t = 0\} \quad \text{and} \quad \{X_t \geq 0\} = \left\{D_t = \bar{\xi}\right\}.
\end{equation*}
The difference between $(D_t)$ and $(X_t)$ is that $(D_t)$ is restricted to take values in $[0,\bar{\xi}]$ only, whereas the process $(X_t)$ can potentially take values on the entire real line.

In the event that demand hits zero or $\bar{\xi}$, we say that the market is in a \textit{negative price regime} (for which negative prices are possible, but not guaranteed) or a \textit{spike regime} and we redefine the electricity price at these points to be given respectively by
\begin{align*}
 b_n\left(x,\mathbf{s}\right)&:= b(0,\mathbf{s}) - \exp(-m_n x) + 1, &&\text{for }  (x,\mathbf{s})\in (-\infty,0] \times \R_+^2,\\
 \text{and}\quad b_s\left(x,\mathbf{s}\right)&:=  b\left(\bar{\xi},\mathbf{s}\right) + \exp\left(m_s\left(x-\bar{\xi}\right)\right) - 1, &&\text{for }  (x,\mathbf{s})\in \left[\bar{\xi},\infty\right) \times \R_+^2.
\end{align*}
Under this extension, the power price expression from Corollary \ref{cor:el_price_two_fuels} is therefore replaced by
\begin{multline*}
 \hat{P}_t := b_n\left(X_t,\mathbf{S}_t\right)\mathbb{I}_{\{0\}}(D_t) + b_{\text{low}}\left(D_t,\mathbf{S}_t\right)\mathbb{I}_{\left(0,\bar{\xi}^i\right]}(D_t)\\
 + b_{\text{mid}}\left(D_t,\mathbf{S}_t\right)\mathbb{I}_{\left(\bar{\xi}^i,\bar{\xi}^j\right]}(D_t) + b_{\text{high}}\left(D_t,\mathbf{S}_t\right)\mathbb{I}_{\left(\bar{\xi}^j,\bar{\xi}\right)}(D_t) + b_s\left(X_t,\mathbf{S}_t\right) \mathbb{I}_{\left\{\bar{\xi}\right\}}(D_t),
\end{multline*}
where $b_{\text{low}}$, $b_{\text{mid}}$, $b_{\text{high}}$ are defined in Corollary \ref{cor:el_price_two_fuels}.

Note that the constants $m_n, m_s>0$ determine how volatile prices are in these two regimes. In such cases, the price of electricity may now be interpreted as being set by a thin tail of miscellaneous bids, which correspond to no particular technology. Therefore, the difference between the electricity price implied by the bid stack and that defined by the negative price or spike regime is independent of fuel prices.

\begin{remark}
It is possible to generate realistic spikes even in the base model (without the inclusion of the spike regime), simply by choosing one of the exponential fuel bid curves to be very steep (large $m_i$). However, this would come at the expense of realistically capturing changes in the merit order, as it artificially stretches the bids associated with that technology.
\end{remark}

Figure \ref{fig:sec2_ppsim} displays the electricity price through time as generated by the stack model for three different choices of parameters, for the same scenario. In Figure \ref{fig:sec2_ppsim2} we show a typical price path in the case that $m_c$, $m_g$, $m_s$ and $m_n$ are very small. This corresponds to a step function bid stack and has been suggested by Aid \emph{et al} \cite{rAid2009}. Clearly, the prices do not exhibit enough variation to match observed time series. The solid line in Figure \ref{fig:sec2_ppsim1} corresponds to more realistic values of $m_c$ and $m_g$; the dashed line illustrates the modification of this path due to the choice of larger values for $m_s$ and $m_n$. Both paths capture the stylized facts of electricity price time series reasonably well. For the purpose of this simulation the prices of coal and gas have been modeled as exponential Ornstein-Uhlenbeck (OU) processes, and demand as an OU process with seasonality truncated at zero and $\bar{\xi}$ (with $(X_t)$ its untruncated version). However, the choice of model for these factors is secondary at this stage, as we are emphasizing the consequences of our choice for the bid stack itself.

\begin{figure}[tbp]
  \centering
  \subfloat[$P_t$ and $\hat{P}_t$ with $m_c,m_g>0$.]{\label{fig:sec2_ppsim1}\includegraphics[width=0.5\textwidth]{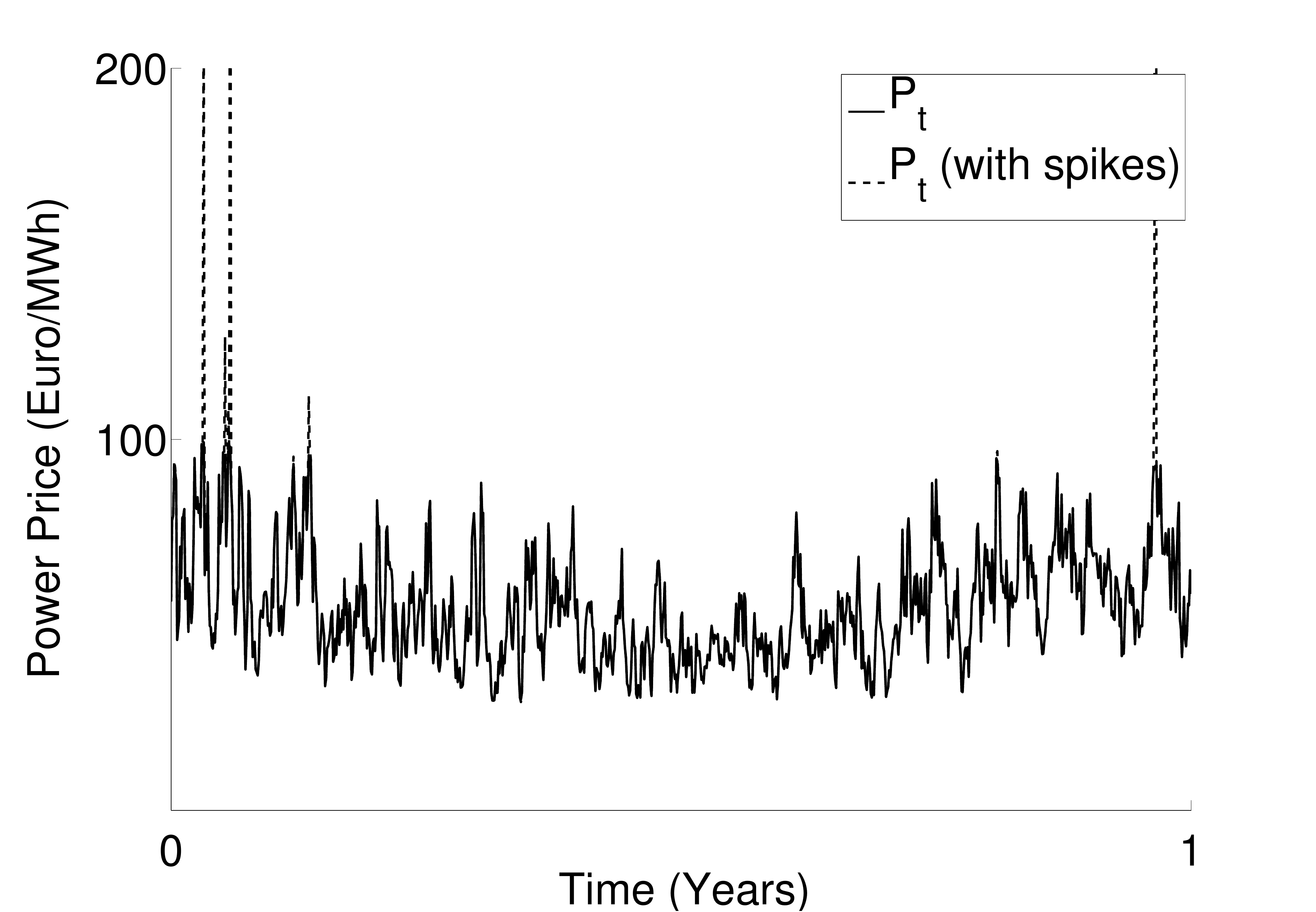}}
  \subfloat[$P_t$ with $m_c,m_g$ small.]{\label{fig:sec2_ppsim2}\includegraphics[width=0.5\textwidth]{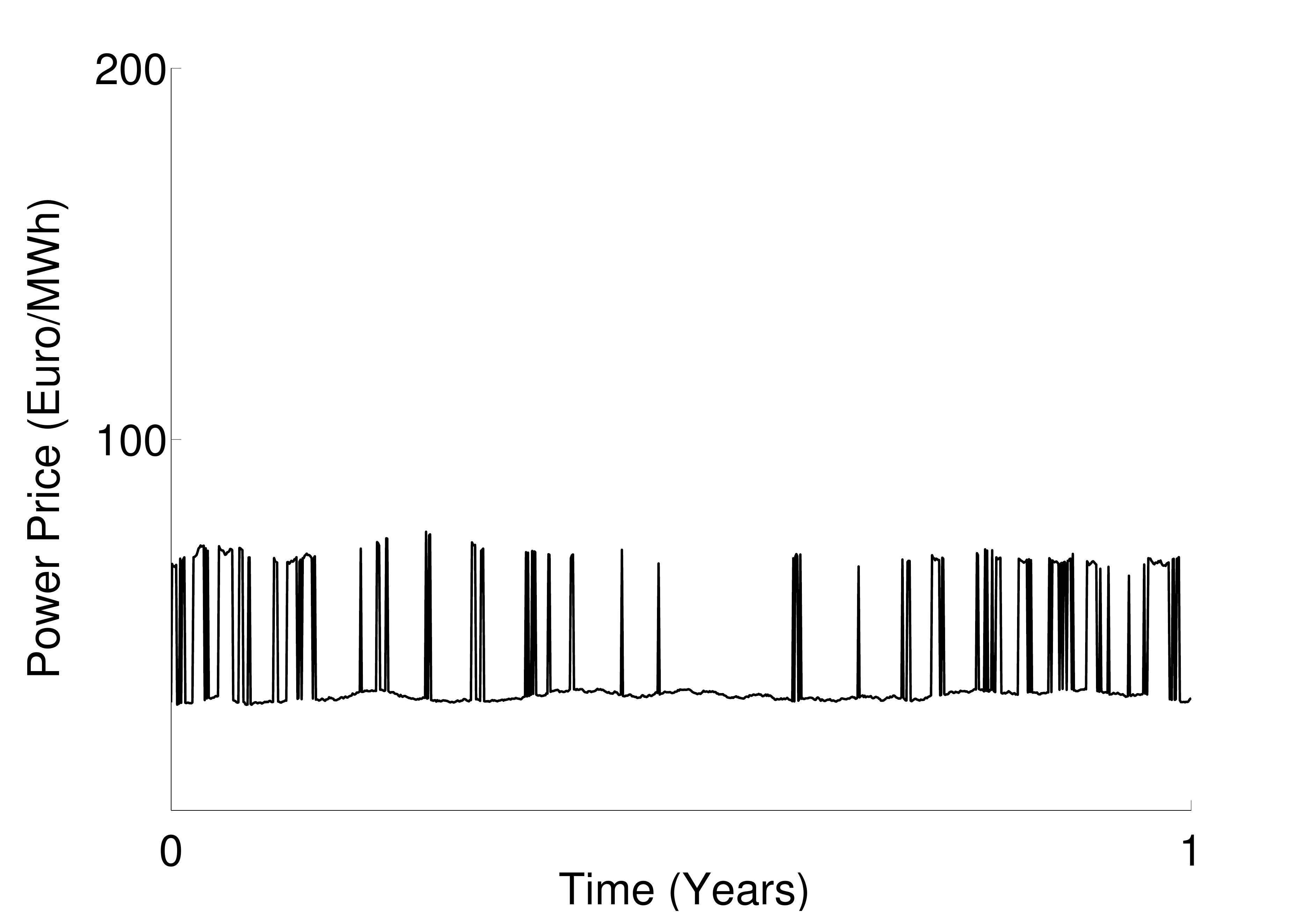}}
  \caption{Simulation of the power price for typical parameters}
 \label{fig:sec2_ppsim}
\end{figure}

\section{Forward Contracts}
\label{str:forwards}
We now turn to the analysis of forward contracts in our structural framework. For the sake of simplicity, we ignore delivery periods and suggest that $T$ be considered as a representative date in a typical monthly delivery period (see \cite{fBenth2008} for more on handling delivery periods). For the purpose of the present discussion, a \textit{forward contract} with maturity $T$ is defined by the payoff
\begin{equation*}
 P_{T} - F_t^p,
\end{equation*}
where $F_t^p$ is the delivery price agreed at the initial date $t$, and paid by the holder of the long position at $T$. Simple arbitrage arguments (cf. \cite{mMusiela2009}) imply that
\begin{equation*}
 F_t^p = \E^\mathbb{Q}\left[P_T|\mathcal{F}_t\right].
\end{equation*}
The result of Corollary \ref{cor:el_price_two_fuels} shows that the payoff of the forward is a function of demand and fuels, so that the electricity forward contract becomes a derivative on fuel prices and demand.
  
\subsection{Closed Form Forward Prices}\label{str:forwards_formulae}
For the explicit calculation of forward prices the following property of Gaussian densities will be useful (see \cite{rGeske1979} for an application of this result to pricing compound options). Let $\varphi_1$ denote the density of the standard univariate Gaussian distribution, and $\Phi_1(\cdot)$ and $\Phi_2(\cdot,\cdot;\rho)$ the cumulative distribution functions (cdfs) of the univariate and bivariate (correlation $\rho$) standard Gaussian distributions respectively.
\begin{lemma}\label{lem:uni_bi_gaussian}
The following relationship holds between $\varphi_1$, $\Phi_1$ and $\Phi_2$:
\begin{multline}\label{eq:uni_bi_gaussian}
\int_{-\infty}^{a}\exp\left(l_1+q_1x\right)\varphi_1(x)\Phi_{1}(l_2+q_2x)\ \Id x
= \exp\left(l_1+\frac{q_1^2}{2}\right)\Phi_2\left(a-q_1,\frac{l_2+q_1q_2}{\sqrt{1+q_2^2}};\frac{-q_2}{\sqrt{1+q_2^2}}\right),
\end{multline}
for all $l_1,l_2,q_1,q_2\in\R$ and $a\in\R\cup\{\infty\}$.
\end{lemma}

\begin{proof}
In equation \eqref{eq:uni_bi_gaussian} combine the explicit exponential term with the one contained in $\varphi_1$ and complete the square. Then, define the change of variable $(x,y)\to(z,w)$ by
 \begin{equation*}
  x = z+q_1,\quad y = w\sqrt{1+q_2^2} + q_2(x-q_1).
 \end{equation*}
The determinant of the Jacobian matrix $J$ associated with this transformation is $|J| = \sqrt{1+q_2^2}$. Performing the change of variable leads to the right hand side of \eqref{eq:uni_bi_gaussian}.
\end{proof}

For the main result in this section we denote by $F^i_t$, $i\in I$, the delivery price of a forward contract on fuel $i$ with maturity $T$ and write $\mathbf{F}_t:=(F^c_t,F^c_t)$.

\begin{proposition}\label{prop:forward}
Given $I=\{c,g\}$, if under $\mathbb{Q}$, the random variables $\log(S^c_T)$ and $\log(S^g_T)$ are jointly Gaussian with means $\mu_c$ and $\mu_g$, variances $\sigma_c^2$ and $\sigma_g^2$ and correlation $\rho$, and if the demand $D_T$ at maturity is independent of $\mathcal{F}^W_T$, then for $t\in[0,T]$, the delivery price of a forward contract on electricity is given by:
 \begin{equation}\label{eq:fwd_price}
 F_t^p = \int_0^{\bar{\xi}^{i_-}} f_{\text{low}}\left(D,\mathbf{F}_t\right)\phi_d(D)\ \Id D
 + \int_{\bar{\xi}^{i_-}}^{\bar{\xi}^{i_+}} f_{\text{mid}}\left(D,\mathbf{F}_t\right)\phi_d(D)\ \Id D + \int_{\bar{\xi}^{i_+}}^{\bar{\xi}} f_{\text{high}}\left(D,\mathbf{F}_t\right)\phi_d(D)\ \Id D,
 \end{equation}
 where $\phi_d$ denotes the density of the random variable $D_T$ and, for $(\xi,\mathbf{x})\in [0,\bar{\xi}]\times \R_+^2$:
\begin{multline*}
f_{\text{low}}\left(\xi,\mathbf{x}\right) = \sum_{i\in I} b_i\left(\xi,x^i\right) \Phi_1\left(R_i(\xi,0)/\sigma\right)\\
+b_{cg}(\xi,\mathbf{x})\exp\left(-\alpha_c\alpha_g\sigma^2/2\right)\left[1-\sum_{i\in I} \Phi_1\left(R_i(\xi,0)/\sigma+\alpha_j\sigma\right)\right],
\end{multline*}
\begin{multline*}
f_{\text{mid}}\left(\xi,\mathbf{x}\right) = b_{i_{+}}\left(\xi-\bar{\xi}^{i_-},x^{i_+}\right)\Phi_1\left(-R_{i_+}\left(\xi-\bar{\xi}^{i_-},\bar{\xi}^{i_-}\right)/\sigma\right) + b_{i_+}\left(\xi,x^{i_+}\right) \Phi_1\left(R_{i_+}(\xi,0)/\sigma\right)\\
+b_{cg}(\xi,\mathbf{x})\exp\left(-\alpha_c\alpha_g\sigma^2/2\right)\left[ \Phi_1\left(R_{i_+}\left(\xi-\bar{\xi}^{i_-},\bar{\xi}^{i_-}\right)/\sigma+\alpha_{i_-}\sigma\right)-\Phi_1\left(R_{i_+}\left(\xi,0\right)/\sigma+\alpha_{i_-}\sigma\right)\right],
\end{multline*}
\begin{multline*}
f_{\text{high}}\left(\xi,\mathbf{x}\right) = \sum_{i\in I} b_i\left(\xi-\bar{\xi}^j,x^i\right)\Phi_1\left(-R_i\left(\xi-\bar{\xi}^j,\bar{\xi}^j\right)/\sigma\right)\\
+b_{cg}(\xi,\mathbf{x})\exp\left(-\alpha_c\alpha_g\sigma^2/2\right)\left[-1+\sum_{i\in I}\Phi_1\left(R_i\left(\xi-\bar{\xi}^j,\bar{\xi}^j\right)/\sigma+\alpha_j\sigma\right)\right],
\end{multline*}
where $j=I\setminus \{i\}$, the constants $\alpha_c,\alpha_g,\beta,\gamma$ are as defined in Corollary \ref{cor:el_price_two_fuels}, and 
\begin{align*}
\sigma^2 &:= \sigma_c^2 - 2\rho\sigma_c\sigma_g+\sigma_g^2,\\
R_i\left(\xi_i,\xi_j\right) &:= k_j+m_j\xi_j-k_i-m_i \xi_i + \log\left(F^j_t\right)-\log\left(F^i_t\right) - \frac12\sigma^2.
\end{align*}
\end{proposition}
\begin{proof}
By iterated conditioning, for $t\in[0,T]$, the price of the electricity forward $F_t^p$ is given by
\begin{equation}\label{eq:el_fwd_cond_exp}
F_t^p:= \E^{\mathbb{Q}}\left[\left.P_T\right|\mathcal{F}_t\right]=\E^{\mathbb{Q}}\left.\left[\left.\E^{\mathbb{Q}}\left[b(D_T,\mathbf{S}_T)\right|\mathcal{F}^0_T \vee \mathcal{F}^W_t\right]\right|\mathcal{F}_t\right].
\end{equation}
The outer expectation can be written as the sum of three integrals corresponding to the cases $D_T\in[0,\bar{\xi}^{i_-}]$, $D_T\in[\bar{\xi}^{i_-},\bar{\xi}^{i_+}]$ and $D_T\in[\bar{\xi}^{i_+},0]$ respectively. We consider the first case and derive the $f_{\text{low}}$ term. The other cases corresponding to $f_{\text{mid}}$ and $f_{\text{high}}$ are proven similarly.

From Corollary \ref{cor:el_price_two_fuels} we know that $b=b_{\text{low}}$ for $D_T\in[0,\bar{\xi}^{i_-}]$. This expression for $P_T$ is easily written in terms of independent standard Gaussian variables $\mathbf{Z}:=(Z_1,Z_2)$ by using the identity
\begin{equation*}
\left( \begin{array}{c}
  \log \left(S^c_T\right)\\
  \log \left(S^g_T\right)
 \end{array}
\right)
=
\left( \begin{array}{c}
  \mu_c\\
  \mu_g
\end{array}
 \right)
+
\left( \begin{array}{cc}
  \sigma_c^2 & \rho\sigma_c\sigma_g \\
  \rho\sigma_c\sigma_g & \sigma_g^2
 \end{array}
 \right)
\left( \begin{array}{c}
 Z_1\\
 Z_2
 \end{array}
 \right).
\end{equation*}
Defining $\hat{d}:=\E^\mathbb{Q}[D_T|\mathcal{F}^0_T]$, the inner expectation can now be written in integral form as
\begin{equation*}
 \E\left[\tilde{b}_{\text{low}}\left(\hat{d},\mathbf{Z}\right)\right] = I_c + I_g + I_{cg},
\end{equation*}
where $\tilde{b}_{\text{low}}(\xi,\mathbf{Z}):=b_{\text{low}}(\xi,\mathbf{S})$ and the expectation is computed with respect to the law of $\mathbf{Z}$. For example, after completing the square in $z_1$,
\begin{equation*}
I_c=\int_{-\infty}^{\infty}\exp\left(l_1+q_1z_2\right)\phi_1(z_2)\Phi_{1}(l_2+q_2z_2)\ \Id z_2,
\end{equation*}
with
\begin{align*}
 l_1&:=\mu_c + k_c + m_c\hat{d}+\frac{\sigma_c^4}{2}, & l_2:&=-\sigma_c^2 - \frac{\mu_c+k_c+m_c\hat{d}+\mu_g}{\sigma_c(\sigma_c-\rho\sigma_g)},\\
 q_1&:=\rho\sigma_c\sigma_g, & q_2&:=\frac{\sigma_g(\sigma_g-\rho\sigma_c)}{\sigma_c(\sigma_c-\rho\sigma_g)}.
\end{align*}
Lemma \ref{lem:uni_bi_gaussian} now applies with $a=\infty$. $I_g$ and $I_{cg}$ are computed similarly.  In all terms, we substitute for $\mu_i$ using the following standard result. For $i\in I$,
\begin{equation}\label{eq:fuel_forward}
 F^i_t = \E^\mathbb{Q}\left[\left. S^i_T\right|\mathcal{F}_t\right] = \exp\left(\mu_i+\frac{1}{2}\sigma_i^2\right), \quad \text{ for } t\in[0,T].
\end{equation}
Substituting the resulting expression for the inner expectation into the outer expectation in \eqref{eq:el_fwd_cond_exp} yields the first term in the proposition.
\end{proof}
\begin{remark}
The assumption of lognormal fuel prices in Proposition \ref{prop:forward} is a very common and natural choice for modeling energy (non-power) prices.  Geometric Brownian Motion (GBM) with constant convenience yield, the classical exponential OU model of Schwartz \cite{eSchwartz1997}, and the two-factor version of Schwartz and Smith \cite{eSchwartz2000} all satisfy the lognormality assumption.
\end{remark}
The above result does not depend upon any assumption on the distribution of the demand at maturity, and as a result, it can easily be computed numerically for any distribution. In markets where reasonably reliable load forecasts exist, one may consider demand to be a deterministic function, in which case the integrals in \eqref{eq:fwd_price} are not needed and the forward price becomes explicit. For cases when load forecasts are not reliable, we introduce another convenient special case below, where demand at maturity has a Gaussian distribution truncated at zero and $\bar{\xi}$. 

To simplify and shorten the notation we introduce the following shorthand notation: 
 \begin{equation*}
\Phi_2^{2\times 1}\left(\left[\begin{array}{c}x_1\\x_2\end{array}\right],y;\rho\right):= \Phi_2(x_1,y;\rho)-  \Phi_2(x_2,y;\rho).
\end{equation*}

\begin{corollary}\label{prop:forwardGaussianD}
In addition to the assumptions in Proposition \ref{prop:forward} let demand at maturity satisfy
\begin{equation*}
 D_T = \max \left(0,\min\left(\bar{\xi},X_T\right)\right),
\end{equation*}
where $X_T\sim N(\mu_d,\sigma_d^2)$ is independent of $\mathbf{S}_T$ under $\mathbb{Q}$. Then for $t\in[0,T]$, the delivery price of a forward contract is given explicitly by
\begin{multline*}
F_t^p=\sum_{i\in I} \exp\left(\frac{m_i^2\sigma_d^2}{2}\right)\left\{b_i\left(\mu_d,F_t^i\right)\Phi_2^{2 \times 1}\left(\left[\begin{array}{c}\frac{\bar{\xi}^i-\mu_d}{\sigma_d}-m_i\sigma_d\\\frac{-\mu_d}{\sigma_d}-m_i\sigma_d\end{array}\right],\frac{R_i(\mu_d,0)-m_i^2\sigma_d^2}{\sigma_{i,d}};\frac{m_i\sigma_d}{\sigma_{i,d}}\right)\right.\\
+\left.b_i\left(\mu_d-\bar{\xi}^j,F^i_t\right)\Phi_2^{2 \times 1}\left(\left[\begin{array}{c}\frac{\bar{\xi}-\mu_d}{\sigma_d}-m_i\sigma_d\\\frac{\bar{\xi}^j-\mu_d}{\sigma_d}-m_i\sigma_d\end{array}\right],\frac{-R_i\left(\mu_d-\bar{\xi}^j,\bar{\xi}^j\right)+m_i^2\sigma_d^2}{\sigma_{i,d}};\frac{-m_i\sigma_d}{\sigma_{i,d}}\right)\right\}\\
+\sum_{i\in I}\delta_i\exp\left(\eta\right)b_{cg}(\mu_d,\mathbf{F}_t) \left\{-\Phi_2^{2 \times 1}\left(\left[\begin{array}{c}\frac{\bar{\xi}^i-\mu_d}{\sigma_d}-\gamma\sigma_d\\\frac{-\mu_d}{\sigma_d}-\gamma\sigma_d\end{array}\right],\frac{R_i(\mu_d,0)+\alpha_j\sigma^2-\gamma m_i\sigma_d^2}{\delta_i\sigma_{i,d}};\frac{m_i\sigma_d}{\delta_i\sigma_{i,d}}\right)\right.\\
+\left.\Phi_2^{2 \times 1}\left(\left[\begin{array}{c}\frac{\bar{\xi}-\mu_d}{\sigma_d}-\gamma\sigma_d\\\frac{\bar{\xi}^j-\mu_d}{\sigma_d}-\gamma\sigma_d\end{array}\right],\frac{R_i\left(\mu_d-\bar{\xi}^j,\bar{\xi}^j\right)+\alpha_j\sigma^2-\gamma m_i\sigma_d^2}{\delta_i\sigma_{i,d}};\frac{m_i\sigma_d}{\delta_i\sigma_{i,d}}\right)\right\}\\
+\Phi_1\left(\frac{-\mu_d}{\sigma_d}\right)\sum_{i\in I} b_i\left(0,F^i_t\right) \Phi_1\left(\frac{R_i(0,0)}{\sigma}\right) +
\Phi_1\left(\frac{\mu_d-\bar{\xi}}{\sigma_d}\right)\sum_{i\in I} b_i\left(\bar{\xi}^i,F^i_t\right)\Phi_1\left(\frac{-R_i\left(\bar{\xi}^i,\bar{\xi}^j\right)}{\sigma}\right),
\end{multline*}
where $j=I\setminus \{i\}$, $\delta_i=(-1)^{\mathbb{I}_{\{i=i_+\}}}	$ and
\begin{equation*}
\sigma_{i,d}^2:=m_i^2\sigma_d^2+\sigma^2 \quad \text{and} \quad \eta:=\frac{\gamma^2\sigma_d^2-\alpha_c\alpha_g\sigma^2}{2}.
\end{equation*}
\end{corollary}

\begin{proof}
We use Lemma \ref{lem:uni_bi_gaussian} with $a<\infty$.  After integrating over demand, each of the terms in $f_{\text{low}}$, $f_{\text{mid}}$ and $f_{\text{high}}$ turns into the difference between two bivariate Gaussian distribution functions. Simplifying the resulting terms leads to the result.
\end{proof}

Although the expression in Corollary \ref{prop:forwardGaussianD} may appear quite involved, each of the terms can be readily identified with one of the five cases listed in Table \ref{tbl:el_price_2fuel}, along with four terms (the last line of $F_t^p$) corresponding to the endpoints of the stack.  Furthermore, it is noteworthy that as compared to Corollary \ref{cor:el_price_two_fuels}, the fuel \emph{forward} prices now replace the fuel \emph{spot} prices in the bid stack curves $b_i$, while $\mu_d$ replaces demand.  The Gaussian cdfs essentially weight these terms according to the probability of the various bid stack permutations. Thus $F_t^p$ can become asymptotically linear in $F_t^c$ or $F_t^g$ when the probability of a single fuel being marginal goes to one.  Finally, we note that a very similar closed-form expression is available for higher moments of $P_T$ and given in Appendix \ref{ap:moments}.  Convenient expressions can also be found for covariances with fuels and for the Greeks (sensitivities with respect to underlying factors or parameters), but are not included.

\begin{remark}\label{rmk:forwards_spikes}
Under the extended model introduced in \textsection \ref{str:spikes} to capture spikes, forward prices are given by the same expression as in Proposition \ref{prop:forwardGaussianD} plus the following simple terms
\begin{multline*}
\exp\left(m_s(\mu_d-\bar{\xi})+\frac12m_s^2\sigma_d^2\right) \Phi_1\left(\frac{\mu_d-\bar{\xi}}{\sigma_d}+m_s\sigma_d\right) -\Phi_1\left(\frac{\mu_d-\bar{\xi}}{\sigma_d}\right)\\
- \exp\left(-m_n\mu_d+\frac12m_n^2\sigma_d^2\right) \Phi_1\left(m_n\sigma_d-\frac{\mu_d}{\sigma_d}\right)+\Phi_1\left(\frac{-\mu_d}{\sigma_d}\right)
\end{multline*}
\end{remark}

\begin{remark}
In most electricity markets, available capacity is often uncertain, due for example to the risk of generator outages.  Since $\bar{\xi}^i$ enters linearly in the exponential function in \eqref{eq:el_price_exp_n} like demand, an extension to stochastic capacity levels should be feasible, though rather involved.  However, if capacity shocks are similar for both fuel types, this additional randomness could more easily be accounted for by adjusting demand parameters $\mu_d$ and $\sigma_d$.  More generally, we note that these parameters could in practice be chosen to calibrate the model to observed power forward (or option) prices, thus using the random variable $D_T$ as a proxy for demand, capacity, and all other non fuel-related risk, along with corresponding risk premia.
\end{remark}

\subsection{Correlation Between Electricity and Fuel Forwards}\label{str:correlations}  
In the American PJM market, coal and gas are the fuel types most likely to be at the margin, with coal historically below gas in the merit order.  Therefore PJM provides a suitable case study for analyzing the dependence structure suggested by our model. In Figure \ref{fig:PJMforwards} we observe the historical co-movement of forward (futures) prices for PJM electricity (both peak and off-peak), Henry Hub natural gas (scaled up by a factor of ten) and Central Appalachian coal. We pick maturities December 2009 and 2011, and plot futures prices over the two years just prior to maturity. Figure \ref{fig:Dec09} covers the period 2007-09, characterized by a peak during the summer of 2008, when most commodities set new record highs.  Gas, coal and power all moved fairly similarly during this two-year period, although the correlation between power and gas forward prices is most striking. 

\begin{figure}[htbp]
  \centering
  \subfloat[Dec 2009 forward price dynamics.]{\label{fig:Dec09}\includegraphics[width=.5\textwidth]{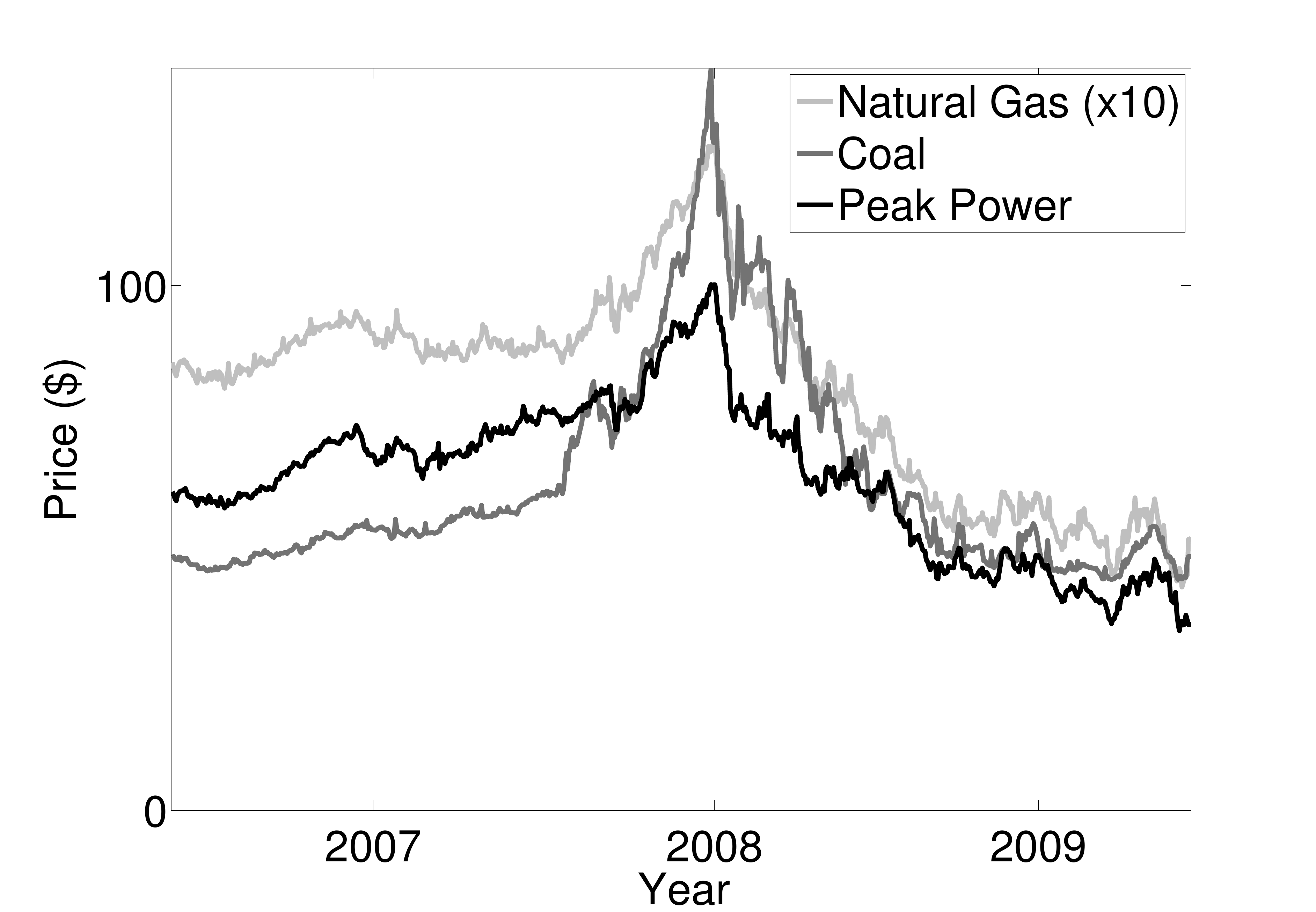}}
  \subfloat[Dec 2011 forward price dynamics.]{\label{fig:Dec11}\includegraphics[width=.5\textwidth]{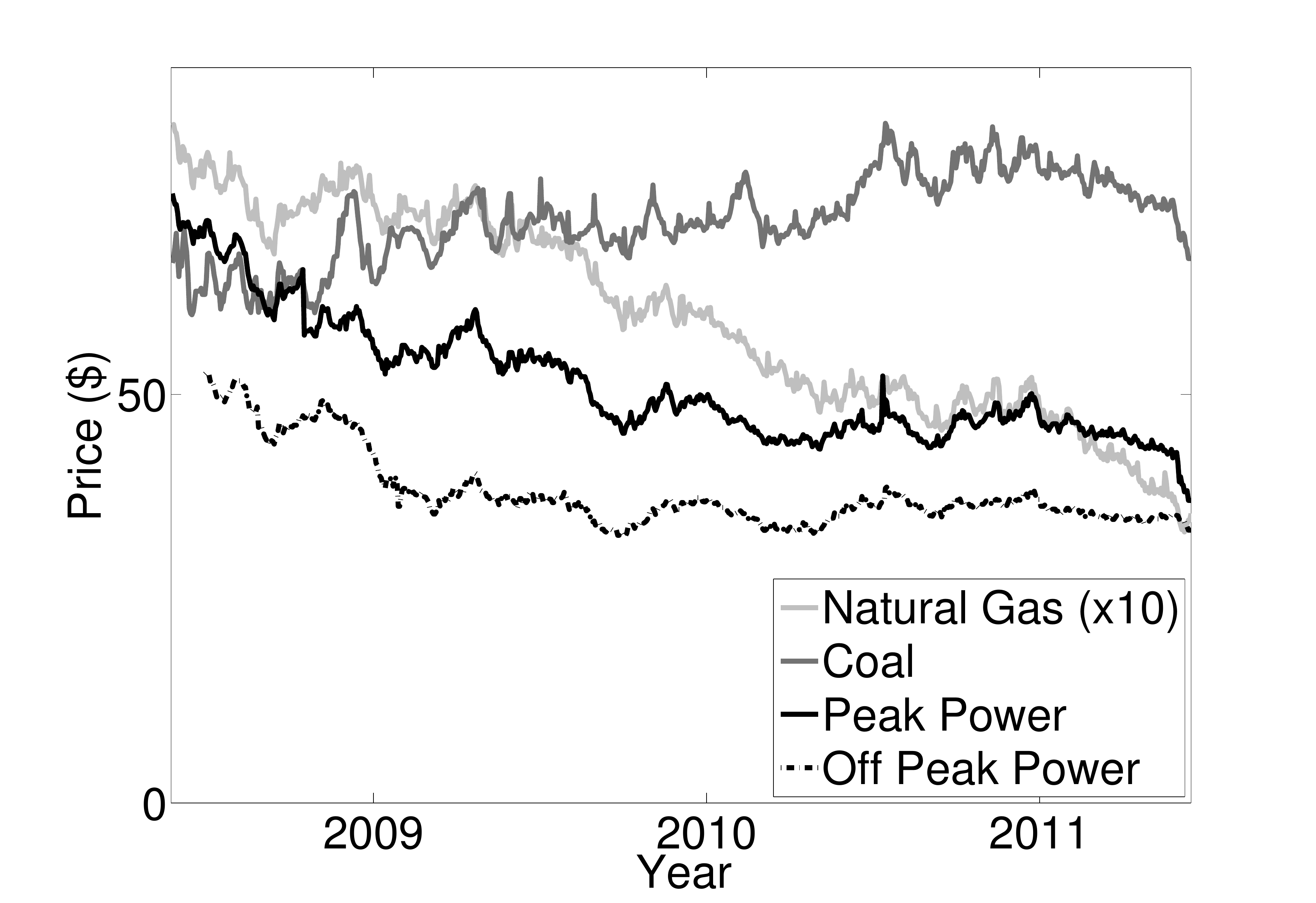}}
 \caption{Comparison of power, gas and coal futures prices for two delivery dates}
 \label{fig:PJMforwards}
\end{figure}
Figure \ref{fig:Dec11} depicts the period 2009-2011, during which, due primarily to shale gas discoveries, gas prices declined steadily through 2010 and 2011, while coal prices held steady and even increased a little.  As a result, this period is more revealing, as it corresponds to a time of gradual change in the merit order.  Our bid stack model implies that the level of power prices should have been impacted both by the strengthening coal price and the falling gas price, leading to a relatively flat power price trajectory.  This is precisely what Figure \ref{fig:Dec11} reveals, with very stable forward power prices during 2010-2011.  The close correlation with gas is still visible, but power prices did not fall as much as gas, as they were supported by the price of coal.  Finally, we can also see that the spread between peak and off-peak forwards for Dec 2011 delivery has narrowed significantly, as we would also expect when there is more overlap between coal and gas bids in the stack.  This subtle change in price dynamics is crucial for many companies exposed to multi-commodity risk, and is one which is very difficult to capture in a typical reduced-form approach, or indeed in a stack model without a flexible merit order and overlapping fuel types.

\section{Spread Options}\label{str:spreads}
This section deals with the pricing of spread options in the structural setting presented above. We are concerned with spread options whose payoff is defined to be the positive part of the difference between the market spot price of electricity and the cost of the amount of fuel needed by a particular power plant to generate one MWh. If coal (gas) is the fuel featured in the payoff then the option is known as a \textit{dark (spark) spread}. Denoting by $h_c,h_g > 0$ the heat rate of coal or gas, dark and spark spread options with maturity $T$ have payoffs
\begin{equation}\label{eq:spread_payoffs}
 \left(P_T - h_cS^c_T\right)^+ \quad \text{and} \quad  \left(P_T - h_gS^g_T\right)^+,
\end{equation}
respectively. We only consider the dark spread but point out that all results in this section apply to spark spreads if one interchanges $c$ and $g$. Further, since spread options are typically traded to hedge physical assets (generating units) the heat rates that feature in the option payoff are usually in line with the efficiency of power plants in the market. Based on the range of market heat rates implied by our stack model, we require\footnote{Explicit formulae for cases of $h_c$ outside of this range are also available, but not included here.} 
\begin{equation}\label{eq:spread_heat_rate_restriction}
 \exp\left(k_c\right) \leq h_c \leq \exp\left(k_c+m_c\bar{\xi}^c\right).
\end{equation}

Then, as usually, the value $(V_t)$ of a dark spread is given by the conditional expectation under the pricing measure of the discounted payoff; i.e.
\begin{equation*}
 V_t = \exp\left(-r(T-t)\right)\E^\mathbb{Q}\left[\left.\left(P_T - h_cS^c_T\right)^{+}\right|\mathcal{F}_t\right],
\end{equation*}
which thanks to Corollary \ref{cor:el_price_two_fuels} is understood to be a derivative written on demand and fuels.
\begin{remark}
While spread option contracts are often written on forwards, we consider spread options on spot prices, as required for our goal of power plant valuation. In addition, as we are interested in closed-form expressions, we limit our attention to the payoffs with strike zero, corresponding to a plant for which fixed operating costs are negligible or relatively small. Including a positive strike, one requires approximation techniques to price a spread option explicitly (such as perturbation of the strike zero case), analogously to approaches proposed for when both commodities are lognormal (cf. \cite{rCarmona2003}).
\end{remark}

\subsection{Closed Form Spread Option Prices}\label{str:spread_price}
The results derived in this section mirror those in \textsection \ref{str:forwards_formulae} derived for the forward contract. Firstly, conditioning on demand, we obtain an explicit formula for the price of the spread. Secondly, we extend this result to give a closed form formula in the case of truncated Gaussian demand.

We keep our earlier notation for the dominant and subordinate technology $i_+$ and $i_-$, and define
\begin{equation*}
\xi^h:=\frac{\log h_c - k_c}{m_c},
\end{equation*}
where $0\leq \xi^h\leq \bar{\xi}^c$. By its definition, $\xi^h$ represents the amount of electricity that can be supplied in the market from coal generators whose heat rate is smaller than or equal to $h_c$.

\begin{proposition}\label{prop:spread}
Given $I=\{c,g\}$, if, under $\mathbb{Q}$, the random variables $\log(S^c_T)$ and $\log(S^g_T)$ are jointly Gaussian distributed with means $\mu_c$ and $\mu_g$, variances $\sigma_c^2$ and $\sigma_g^2$ and correlation $\rho$, then, for $t\in[0,T]$, the price of a dark spread option with maturity $T$ is given by
\begin{multline}\label{eq:V_integrals}
 V_t =\exp\left(-r(T-t)\right)\left\{\int_0^{\min\left(\bar{\xi}^g,\xi^h\right)} v_{\text{low},2}\left(D,\mathbf{F}_t\right) \phi_d(D)\ \Id D\right.\\
 + \int_{\min\left(\bar{\xi}^g,\xi^h\right)}^{\bar{\xi}^{i_-}} v_{\text{low},1}\left(D,\mathbf{F}_t\right) \phi_d(D)\ \Id D
 +\int_{\bar{\xi}^{i_-}}^{\max\left(\bar{\xi}^g,\xi^h\right)} v_{\text{mid},3}\left(D,\mathbf{F}_t\right) \phi_d(D)\ \Id D\\
 +\int_{\max\left(\bar{\xi}^g,\xi^h\right)}^{\min\left(\bar{\xi}^g+\xi^h,\bar{\xi}^c\right)} v_{\text{mid},2,i_+}\left(D,\mathbf{F}_t\right) \phi_d(D)\ \Id D
 +\int_{\min\left(\bar{\xi}^g+\xi^h,\bar{\xi}^c\right)}^{\bar{\xi}^{i_+}} v_{\text{mid},1}\left(D,\mathbf{F}_t\right) \phi_d(D)\ \Id D\\
 + \int_{\bar{\xi}^{i_+}}^{\max\left(\bar{\xi}^c,\bar{\xi}^g+\xi^h\right)} v_{\text{high},2}\left(D,\mathbf{F}_t\right) \phi_d(D)\ \Id D
 + \left. \int_{\max\left(\bar{\xi}^c,\bar{\xi}^g+\xi^h\right)}^{\bar{\xi}} v_{\text{high},1}\left(D,\mathbf{F}_t\right) \phi_d(D)\ \Id D\right\},
 \end{multline}
 where the integrands are given in Appendix \ref{ap:spread_formula} and discussed in some detail.
\end{proposition}

\begin{proof}
As in the proof of Proposition \ref{prop:forward}, by iterated conditioning, for $t\in[0,T]$, the price of the dark spread $V_t$ is given by
\begin{align*}
 V_t :&= \exp\left(-r(T-t)\right)\E^\mathbb{Q}\left[\left.\left(P_T - h_cS^c_T\right)^{+}\right|\mathcal{F}_t\right]\\
    &= \exp\left(-r(T-t)\right)\E^\mathbb{Q}\left[\E^\mathbb{Q}\left[\left.\left(b(D_T,\mathbf{S}_T)-h_cS^c_T\right)^{+}\right|\mathcal{F}^0_T\vee\mathcal{F}^W_t\right]\mathcal{F}_t\right].
\end{align*}
Again we write the outer expectation as the sum of integrals corresponding to the different forms the payoff can take, since the functional form of $b$ is different for $D_T$ lying in the intervals $[0,\bar{\xi}^{i_-}]$, $[\bar{\xi}^{i_-},\bar{\xi}^{i_+}]$, and $[\bar{\xi}^{i_+},\bar{\xi}]$. In addition the functional form of the payoff now depends on whether $D_T\leq \xi^h$ or $D_T\geq \xi^h$ and on the magnitude of $\xi^h$ relative to $\bar{\xi}^c$ and $\bar{\xi}^g$. Therefore, the first case is subdivided into the intervals $[0,\min(\bar{\xi}^g,\xi^h)]$, and $[\min(\bar{\xi}^g,\xi^h),\bar{\xi}^{i_-}]$; the second case is subdivided into $[\bar{\xi}^{i_-},\max(\bar{\xi}^g,\xi^h)]$, $[\max(\bar{\xi}^g,\xi^h),\min(\bar{\xi}^g+\xi^h,\bar{\xi}^c)]$, and $[\min(\bar{\xi}^g+\xi^h,\bar{\xi}^c),\bar{\xi}^{i_+}]$; the third case is subdivided into $[\bar{\xi}^{i_+},\max(\bar{\xi}^c,\bar{\xi}^g+\xi^h)]$, and $[\max(\bar{\xi}^c,\bar{\xi}^g+\xi^h),\bar{\xi}]$.

The integrands $v$ in \eqref{eq:V_integrals} are obtained by calculating the inner expectation for each demand interval listed above, in a similar fashion as in Proposition \ref{prop:forward}. 
\end{proof}

Note that \eqref{eq:V_integrals} requires seven terms in order to cover all possible values of $h_c$ within the range given by \eqref{eq:spread_heat_rate_restriction}, as well as the two cases $c=i_+$ and $c=i_-$.  However, only five of the seven terms appear at once, with only the second or third appearing (depending on $h_c\lessgtr\exp(k_c+m_c\bar{\xi}^g)$) and only the fifth or sixth (depending on $h_c\lessgtr\exp(k_c+m_c(\bar{\xi}^c-\bar{\xi}^g))$).  These conditions can equivalently be written as $\xi^h\lessgtr\bar{\xi}^g$ and $\xi^h\lessgtr\bar{\xi}^c-\bar{\xi}^g$.  Notice that if $c=i_-$, we deduce that $\xi^h<\bar{\xi}^g$ and $\xi^h>\bar{\xi}^c-\bar{\xi}^g$ irrespective of $h_c$, while for $c=i_+$ several cases are possible.

Similar to the analysis of the forward contract earlier, if demand is assumed to be deterministic, then the spread option price is given explicitly by choosing the appropriate integrand from Proposition 
\ref{prop:forward}.  To now obtain a convenient closed-form result for unknown demand, we extend our earlier notational tool for combining Gaussian distribution functions.  For any integer $n$, let 
\begin{equation*}
\Phi_2^{2\times n}\left(\left[\begin{array}{cccc}x_{11}&x_{12}&\cdots &x_{1n}\\x_{21}&x_{22}&\cdots &x_{2n}\end{array}\right],y;\rho\right)= 
\sum_{i=1}^n \left[\Phi_2(x_{1i},y;\rho)-\Phi_2(x_{2i},y;\rho)\right].
\end{equation*}
In addition, we introduce the following notation to capture all the relevant limits of integration. Define the vector $\mathbf{a}:=(a_1,\ldots,a_8)$ by
\begin{equation*}
 \mathbf{a}: =
 \frac1{\sigma_d}\left[
 \left(\begin{array}{cccccccc} 0,  & \bar{\xi}^g\wedge\xi^h, & \bar{\xi}^{i_-}, & \bar{\xi}^g\vee\xi^h, & \bar{\xi}^c\wedge(\bar{\xi}^g+\xi^h), & \bar{\xi}^{i_+}, & \bar{\xi}^c\vee (\bar{\xi}^g+\xi^h), & \bar{\xi}\end{array}\right)- \mu_d \right].
 \end{equation*}
Notice that the components of $\mathbf{a}$ are in increasing order and correspond to the limits of integration in equation \eqref{eq:V_integrals}. In the case that $c=i_+$, all of these values are needed, while the case $c=i_-$ is somewhat simpler because $a_3=a_4$ and $a_5=a_6$ (since by \eqref{eq:spread_heat_rate_restriction}, $\xi^h<\bar{\xi}^c$). However, the result below is valid in both cases as various terms simply drop out in the latter case.
\begin{corollary}\label{prop:spreadGaussianD}
In addition to the assumptions in Proposition \ref{prop:spread} let demand at maturity satisfy
\begin{equation*}
 D_T = \max \left(0,\min\left(\bar{\xi},X_T\right)\right),
\end{equation*}
where $X_T\sim N(\mu_d,\sigma_d^2)$ is independent of $\mathbf{S}_T$. Then for $t\in[0,T]$, the price of a dark spread is given explicitly by
\begin{multline*}
V_t=\exp\left(-r(T-t)\right)\left\{b_c\left(\mu_d,F_t^c\right)\exp\left(\frac{m_c^2\sigma_d^2}{2}\right)\Phi_2^{2 \times 2}\left(\left[\begin{array}{cc}\bar{\xi}^c &a_3\\ a_4 &a_2\end{array}\right],\frac{R_c(\mu_d,0)-m_c^2\sigma_d^2}{\sigma_{c,d}};\frac{m_c\sigma_d}{\sigma_{c,d}}\right)\right.\\
+b_c\left(\mu_d-\bar{\xi}^g,F_t^c\right)\exp\left(\frac{m_c^2\sigma_d^2}{2}\right)\Phi_2^{2 \times 2}\left(\left[\begin{array}{cc}a_8 &a_6 \\ a_7  &a_5\end{array}\right],\frac{-R_c\left(\mu_d-\bar{\xi}^g,\bar{\xi}^g\right)+m_c^2\sigma_d^2}{\sigma_{c,d}};\frac{-m_c\sigma_d}{\sigma_{c,d}}\right)\\
+ b_g\left(\mu_d-\bar{\xi}^c,F_t^g\right)\exp\left(\frac{m_g^2\sigma_d^2}{2}\right)\Phi_2^{2 \times 1}\left(\left[\begin{array}{c}a_8 \\ \bar{\xi}^c\end{array}\right],\frac{-R_g\left(\mu_d-\bar{\xi}^c,\bar{\xi}^c\right)+m_g^2\sigma_d^2}{\sigma_{g,d}};\frac{-m_g\sigma_d}{\sigma_{g,d}}\right)\\
- h_cF_t^c\Phi_2^{2 \times 3}\left(\left[\begin{array}{ccc}a_7 & a_5 &a_3 \\ a_6  &a_4  &a_2 \end{array}\right],\frac{\tilde{R}_c\left(\left(\log h_c-\beta-\gamma \mu_d\right)/\alpha_g\right)}{\sigma_{g,\gamma}};\frac{-\gamma\sigma_d}{\alpha_g\sigma_{g,\gamma}}\right)\\
+b_{cg}\left(\mu_d,\mathbf{F}_t\right)\exp\left(\eta\right)  \left\{ \Phi_2^{2 \times 2}\left(\left[\begin{array}{cc}\bar{\xi}^c &a_3 \\ a_4  &a_2\end{array}\right],\frac{-R_c(\mu_d,0)-\alpha_g\sigma^2+\gamma m_c\sigma_d^2}{\sigma_{c,d}};\frac{-m_c\sigma_d}{\sigma_{c,d}}\right)\right.\\
-\Phi_2^{2 \times 2}\left(\left[\begin{array}{cc}a_8 &a_6 \\ a_7  &a_5\end{array}\right],\frac{-R_c\left(\mu_d-\bar{\xi}^g,\bar{\xi}^g\right)-\alpha_g\sigma^2+\gamma m_c\sigma_d^2}{\sigma_{c,d}};\frac{-m_c\sigma_d}{\sigma_{c,d}}\right)\\
+\Phi_2^{2 \times 1}\left(\left[\begin{array}{c}a_8 \\ \bar{\xi}^c\end{array}\right],\frac{R_g\left(\mu_d-\bar{\xi}^c,\bar{\xi}^c\right)+\alpha_c\sigma^2-\gamma m_g\sigma_d^2}{\sigma_{g,d}};\frac{m_g\sigma_d}{\sigma_{g,d}}\right)\\
-\left.\Phi_2^{2 \times 3}\left(\left[\begin{array}{ccc}a_7 & a_5 &a_3 \\ a_6  &a_4  &a_2\end{array}\right],\frac{-\tilde{R}_c\left((\log H-\beta-\gamma\ \mu_d)/\alpha_g\right)-\alpha_g\sigma^2-\gamma^2\sigma_d^2/\alpha_g}{\sigma_{g,\gamma}};\frac{\gamma\sigma_d}{\alpha_g\sigma_{g,\gamma}}\right)\right\}\\
+\left.\Phi_1\left(-a_8\right)\sum_{i\in I} b_i\left(\bar{\xi}^i,F^i_t\right)\Phi_1\left(\frac{-R_i\left(\bar{\xi}^i,\bar{\xi}^j\right)}{\sigma}\right)-h_cF_t^c\left(1-\Phi_1\left(a_7\right)+\Phi_1(a_6)-\Phi_1(a_5)\right)\right\},
\end{multline*}
 where
 \begin{equation*}
\tilde{R_i}(z) := z + \log\left(F^j_t\right)-\log\left(F^i_t\right) - \frac12\sigma^2 \quad \text{and} \quad \sigma_{i,\gamma}^2:=\gamma^2\sigma_d^2/\alpha_i^2+\sigma^2.
 \end{equation*}
\end{corollary}

\begin{proof}
All terms in \eqref{eq:V_integrals} have the same form as those in Proposition \ref{prop:forward} for forwards: demand appears linearly inside each Gaussian distribution function and in the exponential function multiplying it. Hence, applying Lemma \ref{lem:uni_bi_gaussian} and simplifying lead to the result in the corollary.
\end{proof}

\begin{remark}
Under the extended model introduced in \textsection \ref{str:spikes} to capture spikes, spread prices are given by the same expression as in Proposition \ref{prop:spreadGaussianD} plus the following simple terms\footnote{We require only two of the four extra terms in Remark \ref{rmk:forwards_spikes} due to our assumption on $h_c$ in \eqref{eq:spread_heat_rate_restriction}, which guarantees that for the spike regime, the option is always in the money, while for the negative price regime, it never is.  Hence, if we were to consider put spread options instead of calls, the other two terms would be needed instead.}
\begin{equation*}
\exp\left(m_s(\mu_d-\bar{\xi})+\frac12m_s^2\sigma_d^2\right) \Phi_1\left(\frac{\mu_d-\bar{\xi}}{\sigma_d}+m_s\sigma_d\right) -\Phi_1\left(\frac{\mu_d-\bar{\xi}}{\sigma_d}\right).
\end{equation*}
\end{remark}

\section{Numerical Analysis of Spread Option Prices}
\label{str:pp_valuation}

In this section, we investigate the implications of the two-fuel exponential stack model of \textsection \ref{str:exp_two_fuel} on spread option prices and power plant valuation, as compared to two common alternative approaches.  We analyse prices for various parameter choices and option characteristics, as well as fuel forward curve scenarios.  

Recall that the closed-form spread option prices given by Corollary \ref{prop:spreadGaussianD} required no specification of a stochastic model for fuel prices, but instead imposed only a lognormality condition at maturity $T$.  However, for the purpose of comparing prices across maturities and across modeling approaches, we select a simple example of fuel price dynamics consistent with Corollary \ref{prop:spreadGaussianD}.  We assume coal $(S^c_t)$ and gas $(S^g_t)$ follow correlated exponential OU processes under the measure $\mathbb{Q}$. i.e., for $i \in I$ and $t\in[0,T]$,
\begin{eqnarray}
\label{expOU}
\Id(\log S^i_t) &=& \kappa_i \left(\lambda_i-(\log S^i_t)\right)\Id t + \nu_i \Id W^i_t, \quad S^i_0=s^i
\end{eqnarray}
where $\Id\left<W^c,W^g\right>_t = \varrho \Id t$.  As gas and coal are treated identically in the bid stack model, we consider the symmetric case in which all coal and gas parameters are equal,\footnote{While it is of course not realistic for gas and coal to have identical \emph{prices} (typically differ by a magnitude of about ten and trade in different units), it is plausible that after adjusting for different heat rates, the coal and gas \emph{bids} could indeed coincide and have similar volatility, make our symmetric case less hypothetical.} including both the parameters in \eqref{expOU} and in the exponential fuel bid curves, defined in \eqref{eq:exp_bs}.  All are listed in Table \ref{tbl:parameters}, with the exception of $\varrho$, which we vary throughout our analysis.  All prices are calculated for time $t=0$.  Note that for a given maturity $T$, the parameters $\mu_c$, $\mu_g$, $\sigma_c$, $\sigma_g$, and $\rho$ appearing in \textsection \ref{str:forwards} and \textsection \ref{str:spreads} are related to those in \eqref{expOU} by the following standard results (for $i\in I$):
\begin{eqnarray*}
\mu_i &=& s^i \exp\left(-\kappa_i T\right)+\lambda_i\left(1-\exp\left(-\kappa_i T\right)\right), \\
\sigma_i^2 &=& \frac{\nu_i^2}{2\kappa_i}\left(1-\exp\left(-2\kappa_i T\right)\right),\\
\rho\sigma_c\sigma_g &=&\frac{\varrho\nu_c\nu_g}{\kappa_c+\kappa_g}\left(1-\exp\left(-\kappa_c T-\kappa_g T\right)\right).
\end{eqnarray*}
As for fuel prices, recall that no particular process is required for electricity demand in our model.  Typically driven by temperature, demand is often modelled as rapidly mean-reverting to a seasonal level.  Hence, in our examples, we assume $(X_t)$, as introduced in Corollary \ref{prop:spreadGaussianD}, to be an independent OU process, with a high value for mean-reversion speed (e.g. 100 or more).  For our aim of pricing options with maturities of several months or even years, the values of these parameters are insignificant, as $D_T$ is always well approximated by its stationary distribution.  Hence, in Table \ref{tbl:parameters} we list only the values $\mu_d$ and $\sigma_d$.  We also assume the interest rate $r=0$ throughout.  This completes the base parameter set to be used throughout this section unless otherwise stated. 
\begin{table}[tbp]
 \begin{center}
  \begin{tabular}{ccc|cccc|cc|c}
    \toprule
    \multicolumn{3}{c}{Bid curves}  &\multicolumn{4}{c}{Fuel price processes}  & \multicolumn{2}{c}{Demand} & Rate\\
    \midrule
	$k_i$ & $m_i$ & $\bar{\xi}^i$  &  $\kappa_i$& $\nu_i$ & $\lambda_i$ & $s^i$  & $\mu_d$ &$\sigma_d$ & r\\
    2 & 1 & 0.5 & 1 & 0.5 & $\log(10)$ & $10$ & 0.5 & 0.2 & 0 \\
   \bottomrule
  \end{tabular}
 \end{center}
 \caption{Parameters used throughout \textsection \ref{str:pp_valuation} (for $i\in I$)}  \label{tbl:parameters}
\end{table}

Next, we introduce three scenarios designed to assess the role of observed fuel forward curves and corresponding implications for bid stack structure.  Fuel forward curves reveal crucial information about the probability of future merit order changes.  For example, if bids from coal and gas are currently at similar levels (as in our parameter set) but one fuel is in backwardation (forwards decreasing in $T$), while the other is in contango (increasing in $T$), then the future dynamics of power prices (under $\mathbb{Q}$) should reflect the high chance of the coal and gas bids separating.  We compare the following scenarios:
\begin{enumerate}[(I)]
\item No fuel forward inputs (forward prices implied from \eqref{expOU})
\item Gas in contango; coal in backwardation (linear with stepsize of 0.2 per month)
\item No fuel forwards, but gas bids above coal ($\lambda_c=\log(s^c)=\log(7)$, $\lambda_g=\log(s^g)=\log(13)$)
\end{enumerate}
Note that in Scenario II, fuel forward curves are inputs assumed to be observed from the market, instead of being generated by a model. The standard approach to resolving the inconsistency between the market and the model in \eqref{expOU} is to calibrate to fuel forwards for each $T$ via a shift in the mean level $\mu_i$ (or more formally via a time dependent long-term mean $\lambda_i$).  

\subsection{Spread Option Price Comparison}\label{str:spread_comp}

To test our model's prices for spark and dark spread options, we compare with two other typical approaches to spread option pricing: Margrabe's formula (cf. \cite{wMargrabe1978}) and a simple cointegration model (cf. \cite{gEmery2002}, \cite{cJong2009} for discussions of cointegration between electricity and fuels).   

\subsubsection{Margrabe's Formula}\label{str:margrabe}
Assume that under the measure $\mathbb{Q}$, the electricity price $P_T$ and fuel price $S^i_T$, $i\in I$, are jointly lognormal, with correlation $\rho_{p,i}$.  Writing $\mu_p$ and $\sigma_p^2$ for the mean and variance of $\log(P_T)$, then for $t\in[0,T]$, the price of a spread option with payoff \eqref{eq:spread_payoffs} is given by
\begin{equation*}\label{eq:Margrabe}
V^m_t=\exp\left(-r(T-t)\right) \left[ F^p_t\Phi_1\left(\frac{\log\left(F_t^p/h_iF_t^i\right)+\sigma_{p,i}^2/2}{\sigma_{p,i}}\right)- h_iF^i_t\Phi_1\left(\frac{\log\left(F_t^p/h_iF_t^i\right)-\sigma_{p,i}^2/2}{\sigma_{p,i}}\right)\right],
\end{equation*}
where $\sigma_{p,i}^2=\sigma_p^2-2\rho_{p,i}\sigma_p\sigma_i+\sigma_i^2$, and all other notation is as before.

\subsubsection{Cointegration Model}
Let $Y_T$ be an independent Gaussian random variable under $\mathbb{Q}$, with mean $\mu_y$ and variance $\sigma_y^2$.  Then for constant weights $w_c,w_g>0$ (the cointegrating vector), we define $P_T$ by
\begin{equation*}
 P_T:=w_cS^c_T + w_gS^g_T + Y_T.
\end{equation*}
No closed form results are available for spread options, so prices are determined by simulation.

\subsubsection{Comments on Comparison Methodology}\label{str:comments}

In order to achieve a sensible comparison between the stack model and either Margrabe or the cointegration model, the mean and variance of $P_T$ should be chosen appropriately, for each maturity.  For a single fixed $T$ this simply requires choosing parameters $\mu_p$ and $\sigma_p$ (in the case of Margrabe) or $\mu_y$ and $\sigma_y$ (in the case of the cointegration model) to exactly match the mean and the variance of $P_T$ produced by the stack model.  As we shall often compare models simultaneously across many maturities, we use a variation of this idea, finding the best fit of an OU process for $\log(P_t)$ or $(Y_t)$.  For Margrabe,
\begin{equation*}
\Id(\log P_t) = \kappa_p \left(\lambda_p-(\log P_t)\right)\Id t + \nu_p \Id W^p_t, \qquad  P_0=p_0, 
\end{equation*}
where $\Id\left<W^p,W^i\right>_t = \varrho_{p,i} \Id t$ for $i \in I$, while for the cointegration case $(Y_t)$ is an independent OU process.  In all cases we utilize the closed form expression for the variance of $P_T$ in the stack model, given in Appendix \ref{ap:moments}.  

Finally, to reflect the symmetry between coal and gas in our parameter set, in the cointegration model we set $w_c = w_g= 1/2\exp(k_c+m_c(\mu_d\bar{\xi}^c))$, such that the power price has equal dependence on each fuel, with a price level linked to the most likely marginal bid levels of coal and gas.  Note that since both underlying fuels appear in the cointegration model like in the stack model, a comparison across values of fuel price correlation $\varrho$ is quite natural.  On the other hand, the correlation parameter which we vary in the Margrabe approach is $\varrho_{p,i}$, which correlates electricity and fuel.  Hence the choices of $\varrho$ or $\varrho_{p,i}$ in the following plots are not perfect comparisons, but rather illustrate the role of correlation in generating a range of prices for each model.  Due to variations across maturities, there is no direct link between a chosen value $\varrho$ and an appropriate $\varrho_{p,i}$ short of estimating these from data.
 
\begin{figure}[htbp]
  \centering
  \subfloat[All models (Scenario I)]{\label{CompH1}\includegraphics[width=0.5\textwidth]{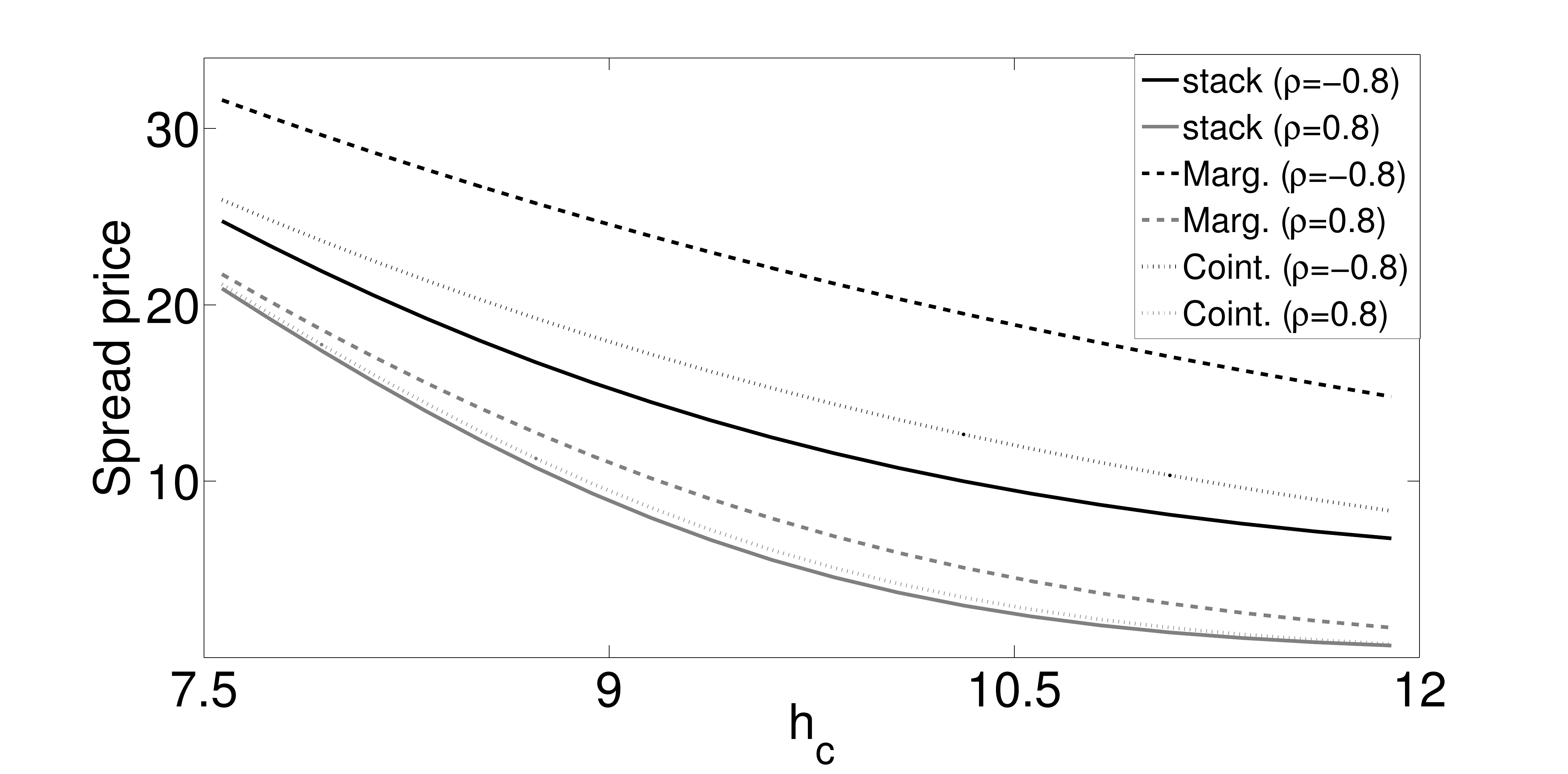}}
  \subfloat[Stack model ($\mu_d=0.3$,$\sigma_d=0.12$)]{\label{CompH2}\includegraphics[width=0.5\textwidth]{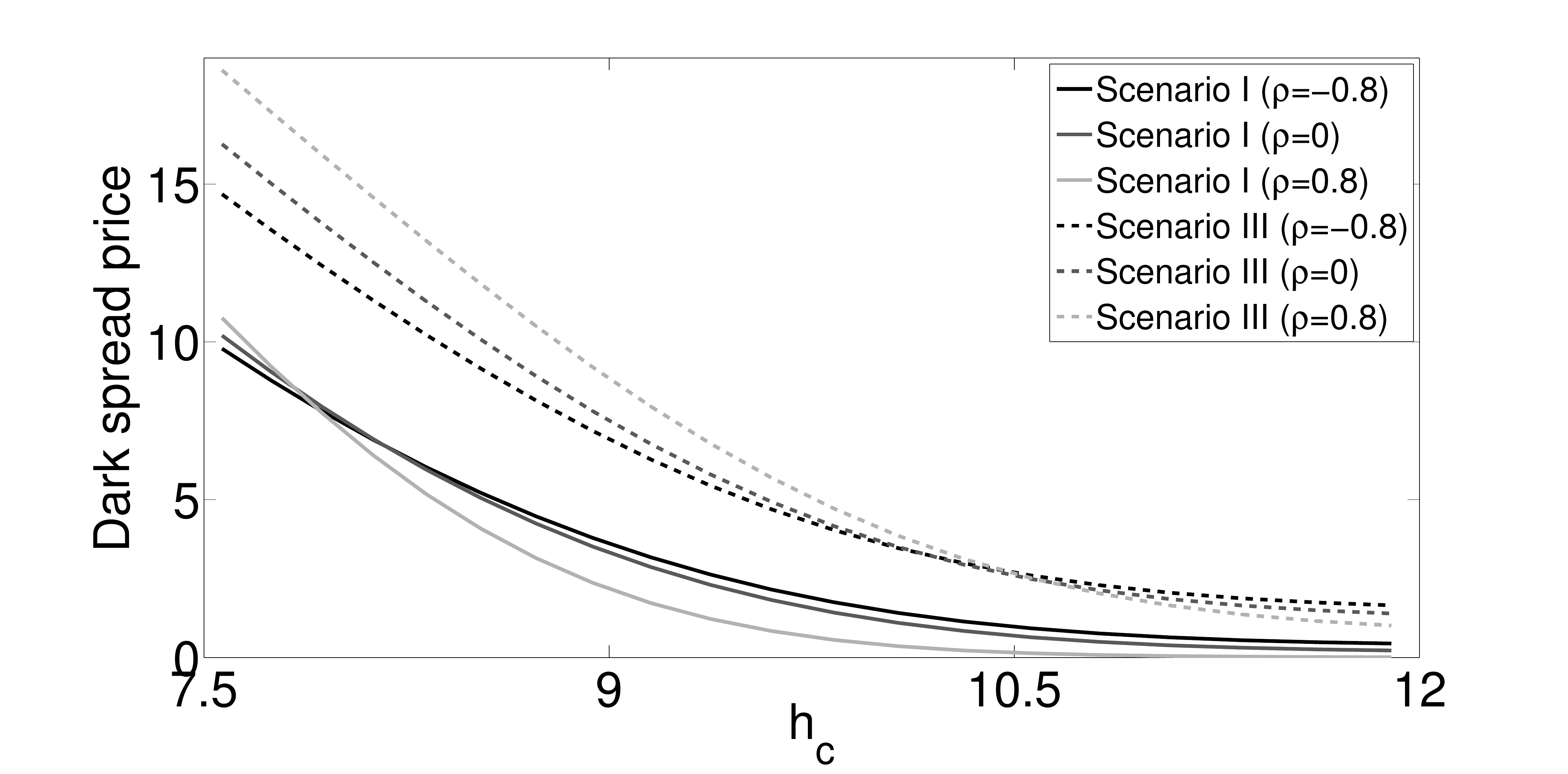}}
   \caption{Spread option prices against $h_i$ for different correlations and demand levels}
  \label{ComparisonH}
 \end{figure}

\subsection{Spread Option Parameter Analysis}\label{str:parameter_analysis}
\subsubsection{Spread prices versus heat rate $h_i$}
For a fixed maturity $T=1$, we plot dark spread option prices as a function of heat rate $h_c$ (over a range corresponding to \eqref{eq:spread_heat_rate_restriction}).  Figure \ref{CompH1} illustrates the Scenario I results for each of the three models considered, for two different correlation parameters, $\varrho=\pm 0.8$. In all three models, negative correlation logically raises the option price, as it increases the volatility of the spread.  The stack model generally predicts lower spread option prices than Margrabe, and a smaller gap between different correlation levels.  This is because the strong structural link keeps long-term levels of power and gas close together, thus narrowing the spread distribution relative to the weaker case of correlated Brownian Motions.  The gap between Margrabe and the stack model widens with negative correlation, while for extreme positive correlation (not plotted), Margrabe can underprice the stack model.  Another way of understanding this phenomenon is to notice that the bid stack automatically imposes a positive dependence structure between electricity and its underlying fuels, which can only be somewhat weakened or strengthened by varying fuel price correlation through $\varrho$.  The cointegration approach shares this characteristic, and therefore prices much closer to the stack model than Margrabe, but still somewhat higher for $\varrho=0.8$.    

Unlike for Margrabe, Figure \ref{CompH2} reveals interestingly that the dependence on correlation $\varrho$ does not hold strictly for the stack model, as exceptions can be found.  Considering a market with very low demand ($\mu_d=0.3,\sigma_d=0.12)$, we observe that the relationship with $\varrho$ is reversed for low values of $h_c$.  Since demand is so low, the price is typically set by the cheaper fuel and hence only in the money for coal bids generally below gas.  Hence for negative $\varrho$, we do not receive the typical high payoffs from high gas and low coal states, and get zero payoff when coal moves above gas in the stack.  Figure \ref{CompH2} also confirms that this reversal is most pronounced if coal bids are shifted to be below gas bids, as described by Scenario III. 

\begin{figure}[htbp]
     \centering
  \subfloat[Scenario I: Spark or dark spread]{\label{CompT1}\includegraphics[width=0.5\textwidth]{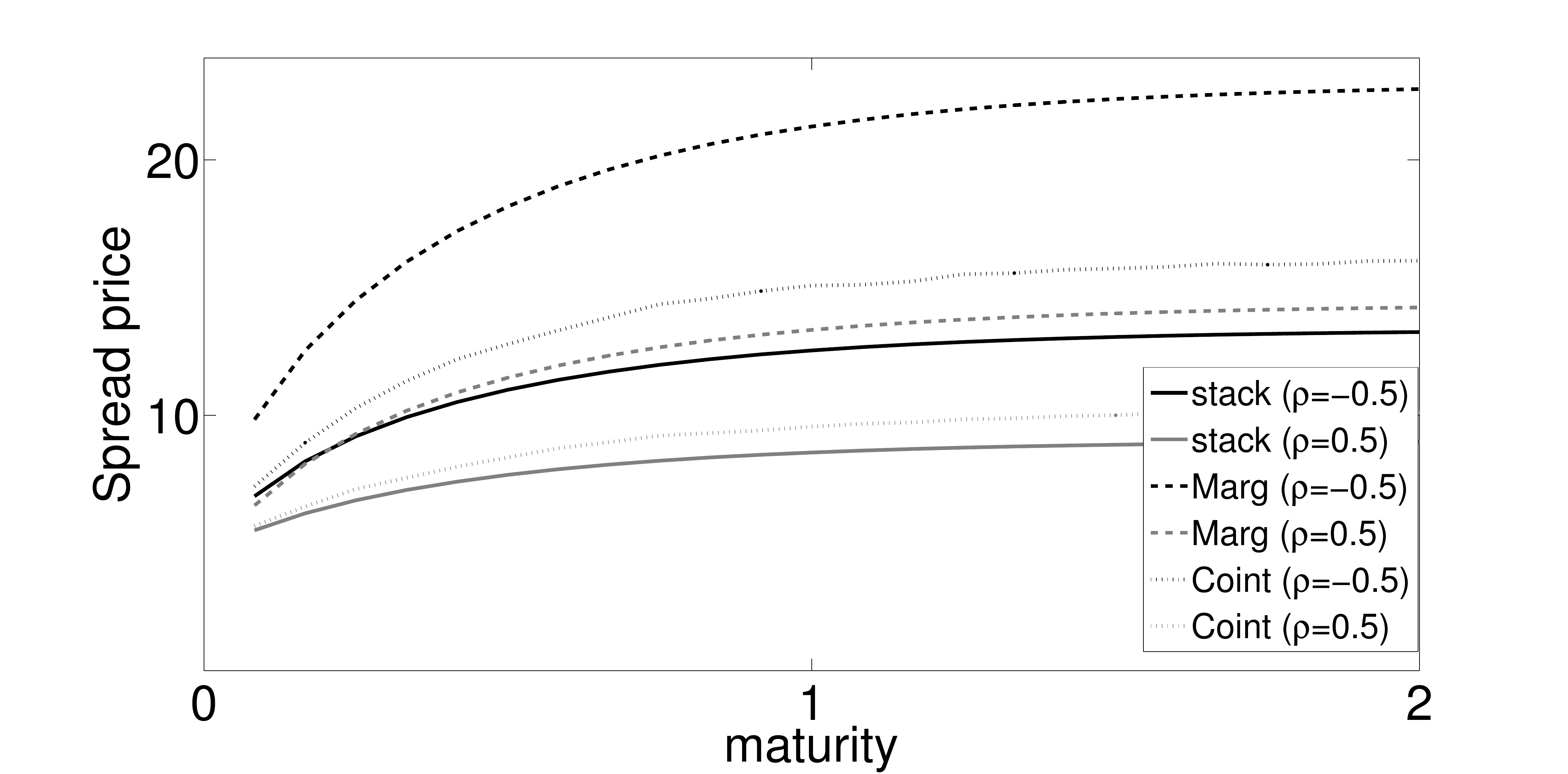}}
  \subfloat[Scenario II: Spark spread]{\label{CompT2}\includegraphics[width=0.5\textwidth]{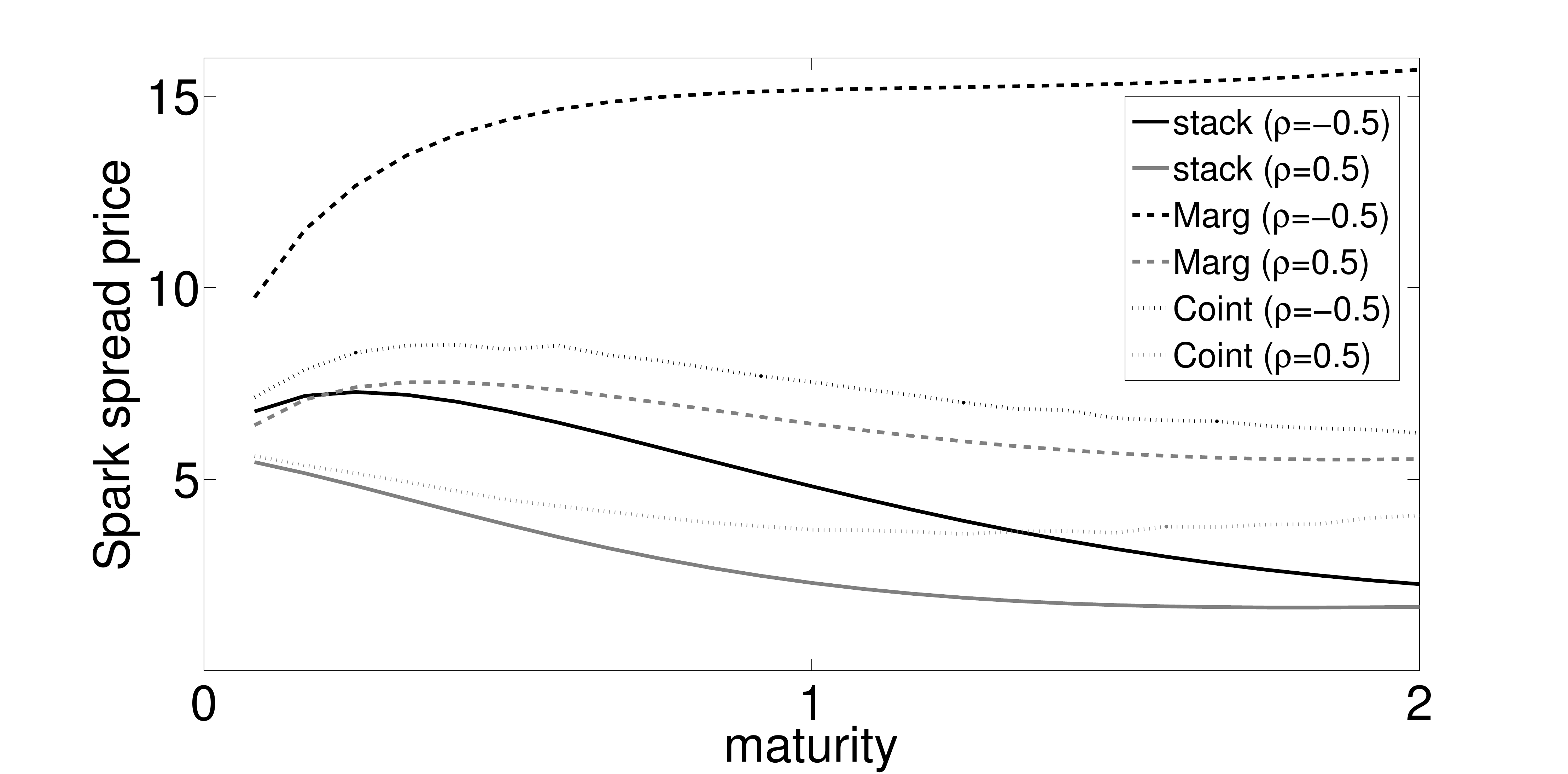}}\\
  \caption{Option prices against $T$ for different correlations and fuel forward scenarios}
\label{ComparisonT}
\end{figure}

\subsubsection{Spread prices versus maturity $T$}
We next investigate spread option prices against maturity $T$, and for this purpose fix $h_i=\exp(k_i+m_i\bar{\xi}^i/2)$ in option payoffs, matching the median heat rate in the market.  We again compare several correlation levels for all three models, and now include Scenario II to test the impact of fuel forward curves.  In Figure \ref{CompT1} (Scenario I), spread options are typically increasing in maturity as expected, and flatten out as the price processes approach their stationary distributions, with an ordering of the three models resembling Figure \ref{CompH1}.  More interestingly, in Scenario II (Figure \ref{CompT2}), longer term spark spread options drop significantly in value in stack model, are thus greatly overpriced by Margrabe, and significantly overpriced by the cointegration model as well.  As the gas forward curve is now in contango, while coal is in backwardation, coal will almost always be below gas in the future bid stack, especially for very long maturities.  Hence, a spark spread option has relatively little chance of being in the money, as this would require unusually high demand.  This is a good example of a dependency which cannot be captured by Margrabe or other reduced-form models, but is automatically captured by the merit order built into the stack model.  Moreover, fuel forward prices are direct inputs into our expressions for spread options, avoiding the need for an additional calibration step to first match observed fuel forwards, as is the case for the other approaches.

\begin{figure}[htbp]
  \centering
  \subfloat[Scenario II: Spark Spread Comparison]{\label{CompT1var}\includegraphics[width=0.5\textwidth]{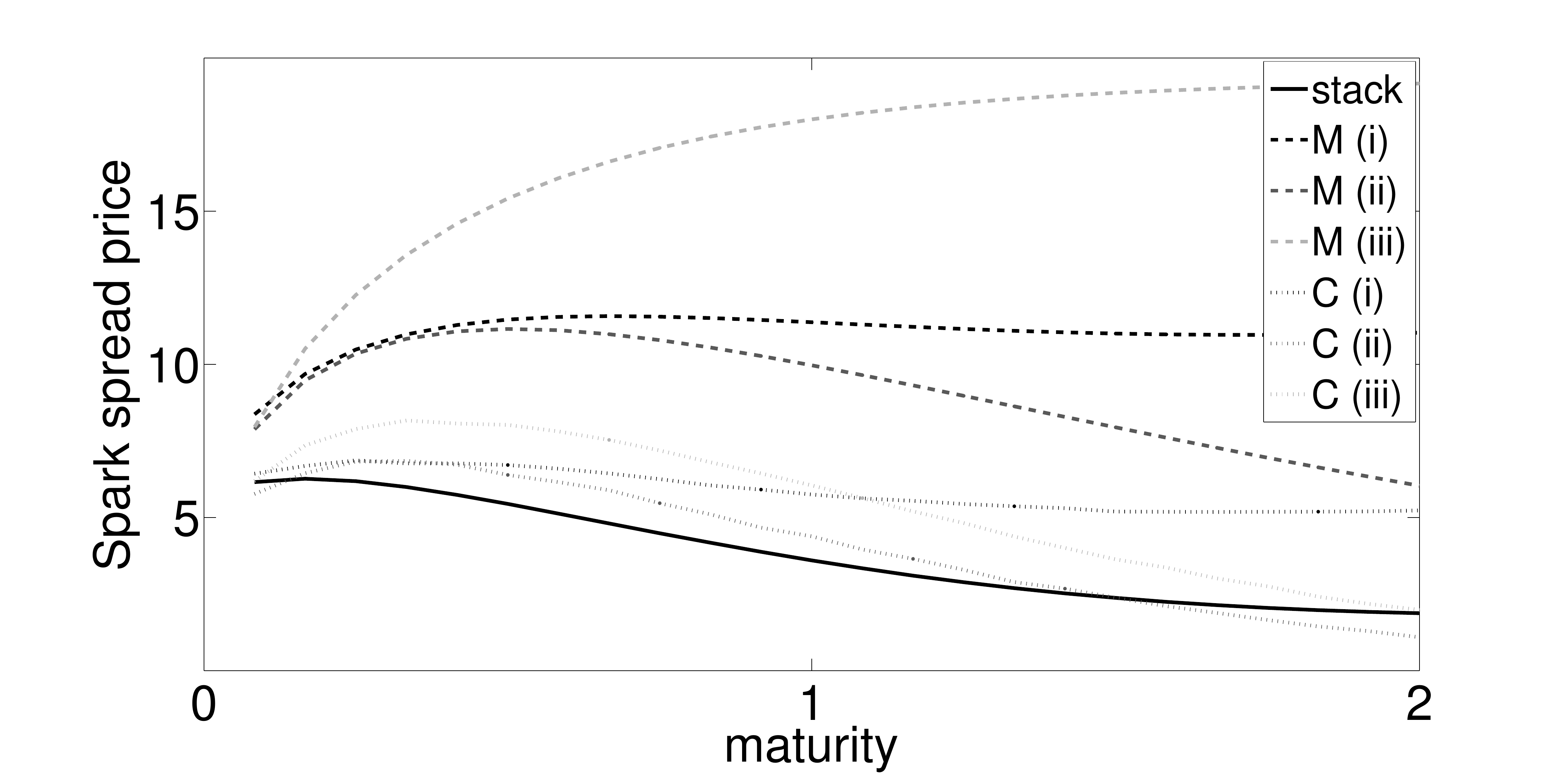}}
  \subfloat[Comparison of variances]{\label{CompT2var}\includegraphics[width=0.5\textwidth]{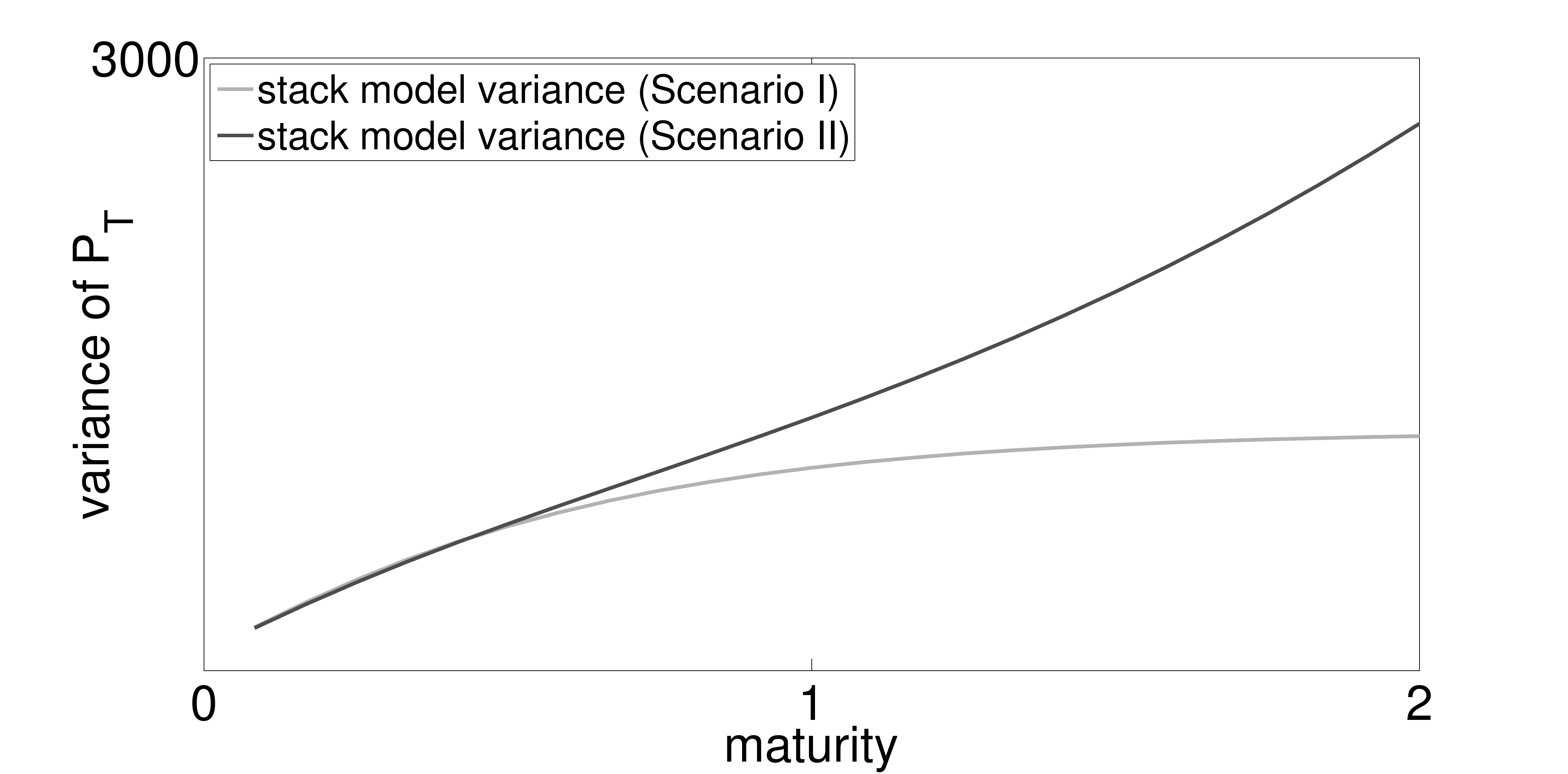}}
   \caption{Analysis of impact of matching mean and variance of $P_T$ for Scenario II}
 \label{ComparisonTvariances}
 \end{figure}
 
So far all plots have assumed that both the mean and variance of the power price distributions are matched in all three models, via the procedure described in Section \ref{str:comments}.  One might question whether this is realistic.  In practice, we only have history (and possibly observed forward curves) to calibrate each model, and thus should not be borrowing extra information about the future from the stack model's structure when calibrating the other approaches.  While matching the mean is reasonable as it is analogous to matching observed power forwards, matching the variance is less justifiable.  In Figure \ref{CompT1var}, we compare spark spread option prices for Margrabe and the cointegration model in Scenario II (as in Figure \ref{CompT2} but now $\rho=0$) for three different calibration assumptions:  full matching as earlier; matching means but not variances; matching neither means nor variances. Here `not matched' implies that means and/or variances are instead fitted to Scenario I levels (a proxy for history).  We note significant differences between all cases.  Failing to match the mean implies a greater overpricing of spark spreads in this case, while failing to match the variances acts in the opposite direction here, lowering the price since the forward looking variance (implied by the stack in Scenario II) is higher than the variance in Scenario I (see Figure \ref{CompT2var}).  While other scenarios could lead to different patterns, it is clear that significant price differences can occur due to the likely changes in the merit order.  In Margrabe, no information is transmitted from fuel forward prices to the distribution of $P_T$, while in the cointegration model limited information is transmitted, since the relative dependence on coal and gas is fixed initially by $w_c,w_g$, instead of dynamically adapting to fuel price movements (and demand).  In contrast, the stack model produces highly state-dependent power price volatility and correlations reflecting known information about the future market structure.  

\begin{figure}[htbp]
  \centering
  \subfloat[Implied correlation varying $\bar{\xi}^g$]{\label{ImpCorrXiG}\includegraphics[width=0.5\textwidth]{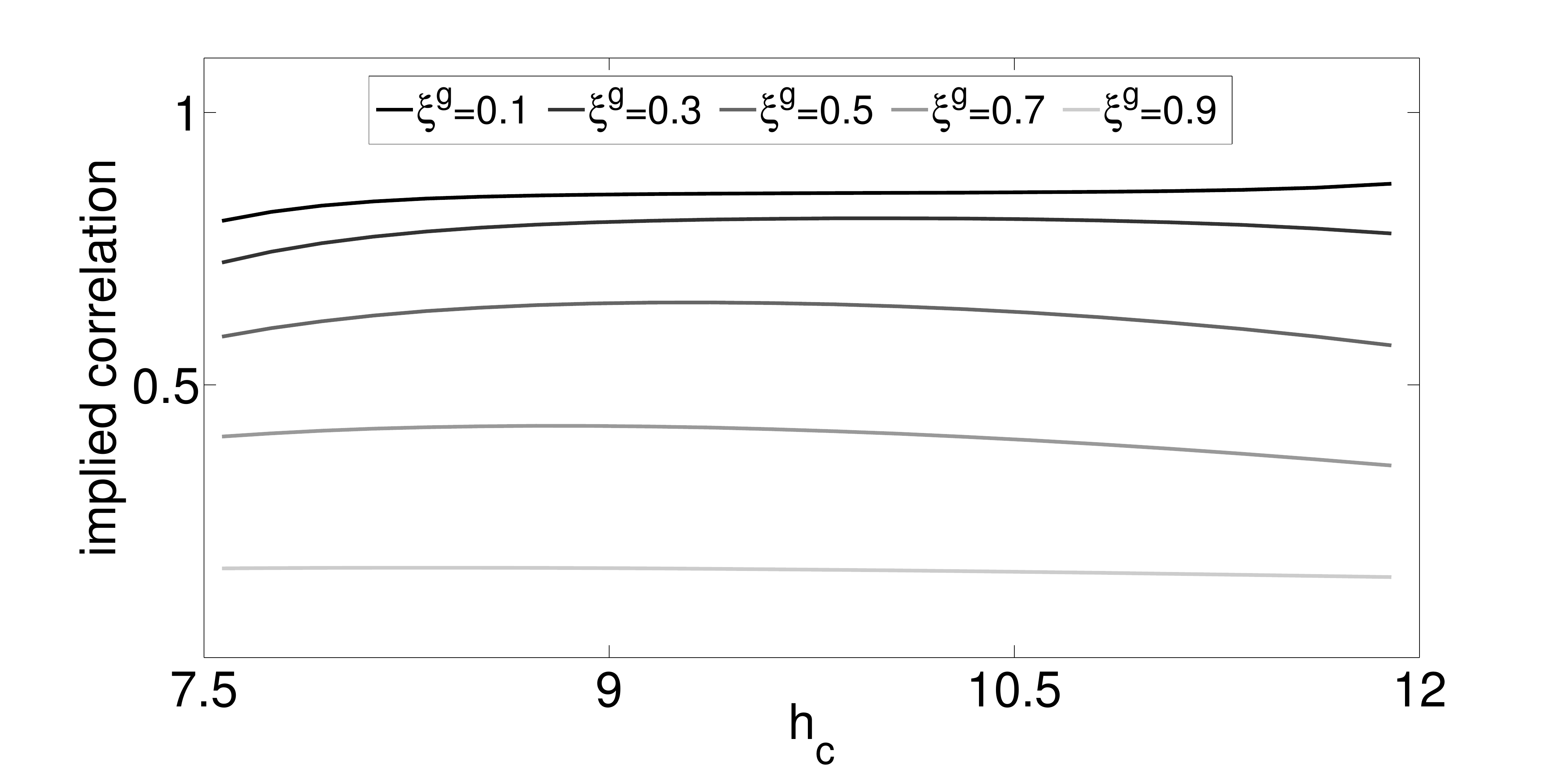}}
  \subfloat[Implied correlation varying $\mu_d$]{\label{ImpCorrMuD}\includegraphics[width=0.5\textwidth]{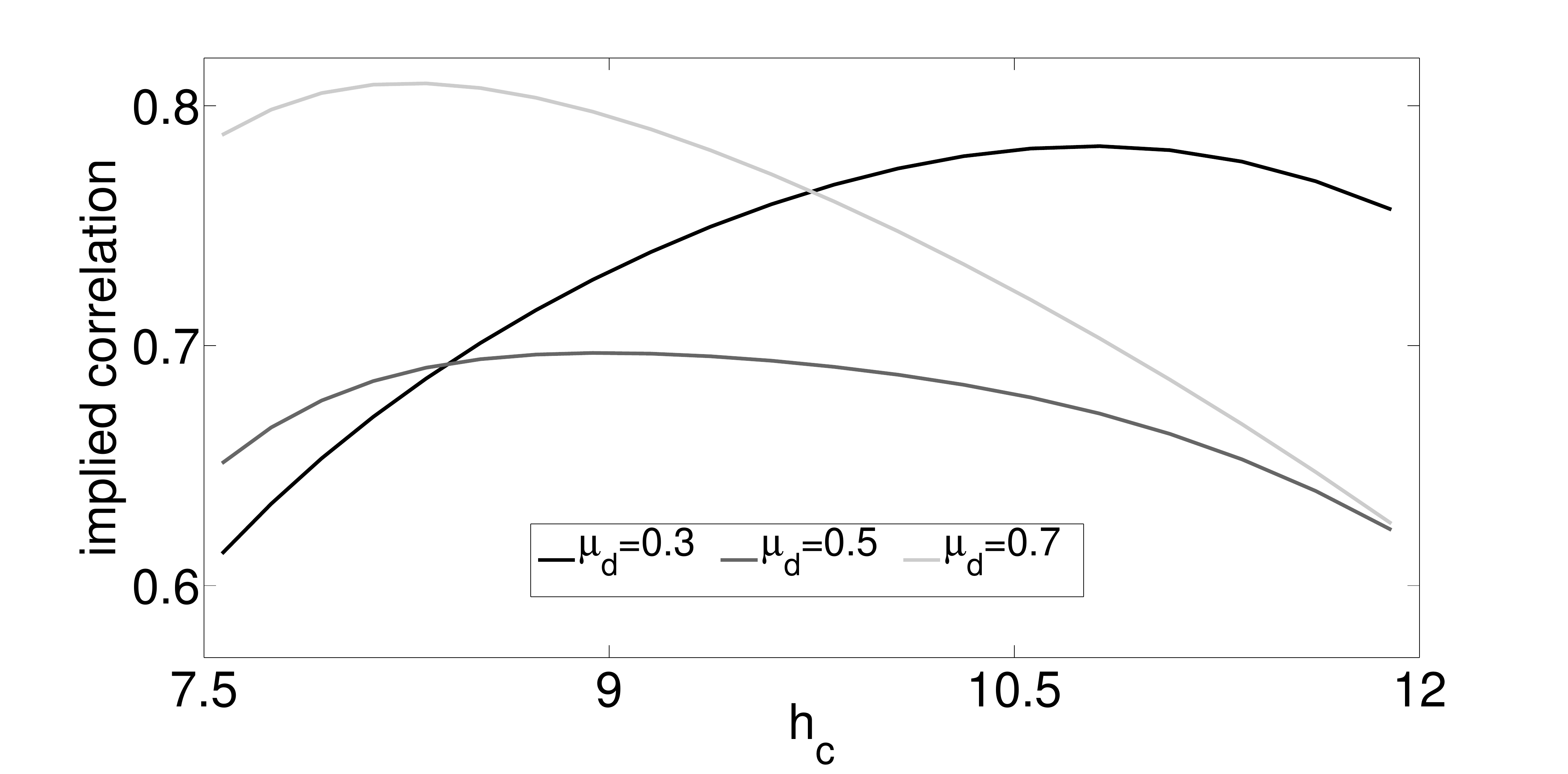}}\\
   \subfloat[$\varrho^{imp}$ surface (Scenario I)]{\label{ImpCorrSurf1}\includegraphics[width=0.5\textwidth, height=0.2\textheight]{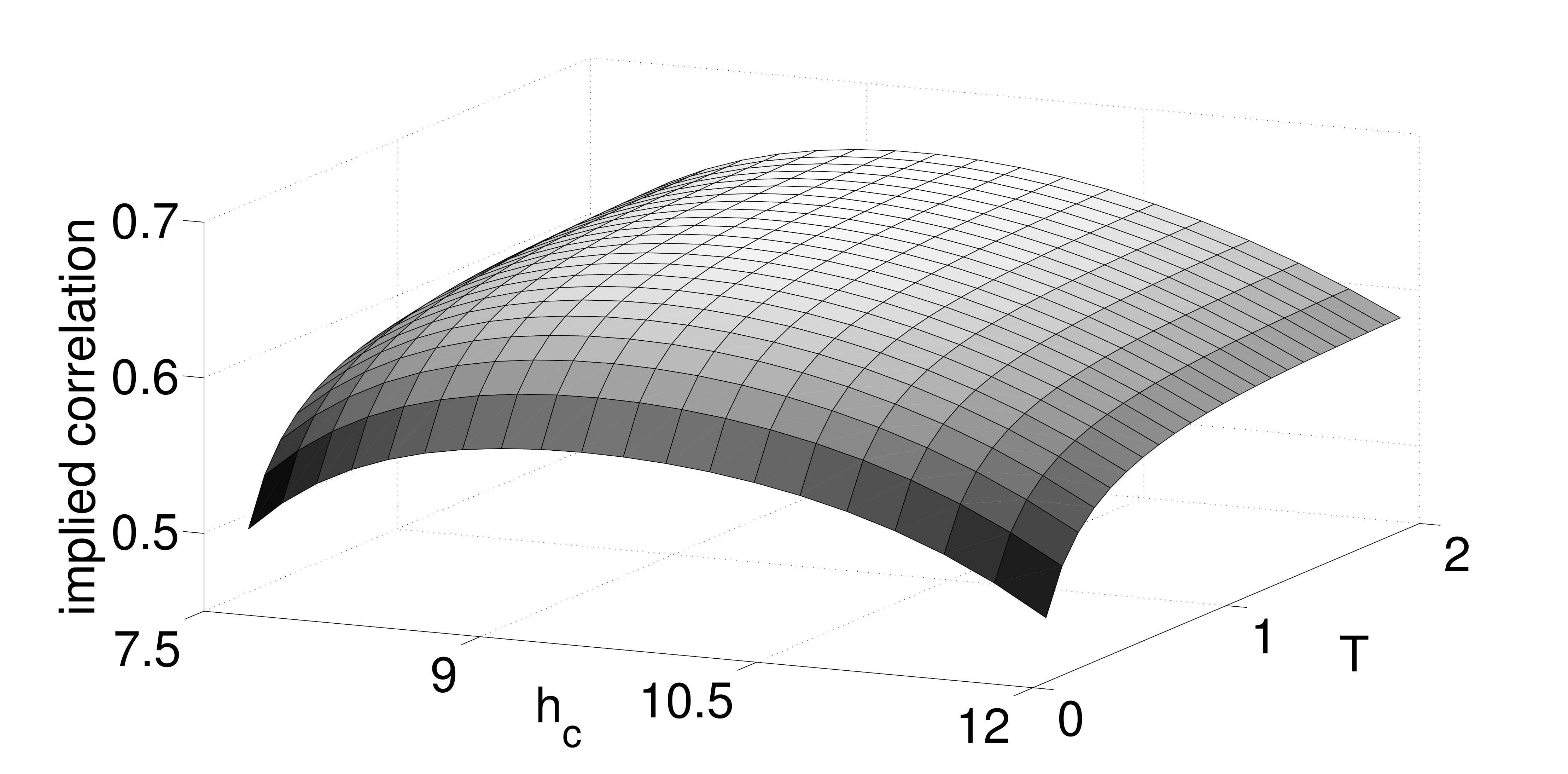}}
  \subfloat[$\varrho^{imp}$ surface (Scenario II)]{\label{ImpCorrSurf2}\includegraphics[width=0.5\textwidth, height=0.2\textheight]{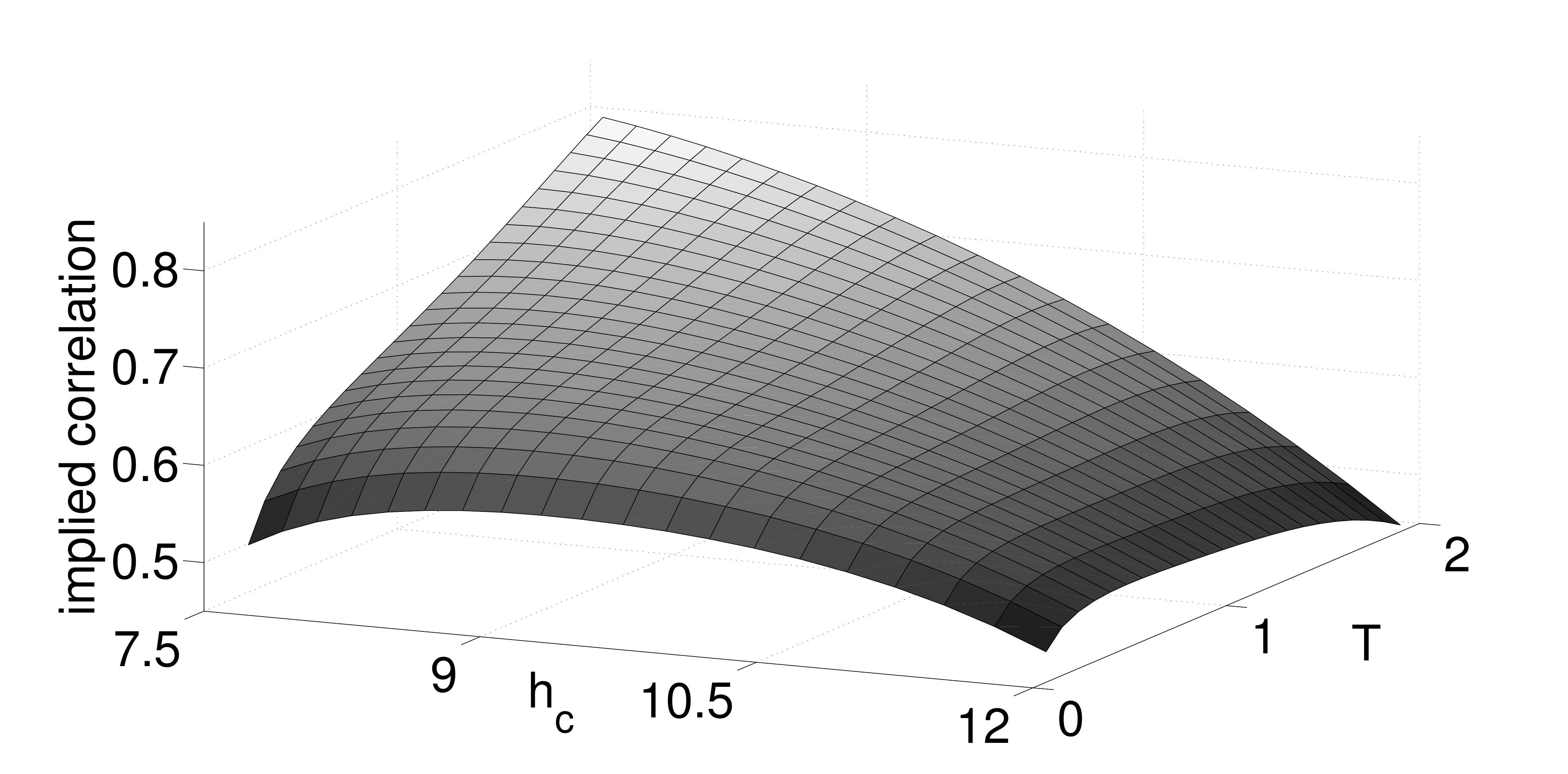}}
  \caption{Implied correlation analysis for various parameters and scenarios}
  \label{ImpCorrel}
 \end{figure}
 
\subsection{Implied Correlation Analysis}\label{str:imp_correl}

We next analyse `implied correlation' $\varrho^{imp}_{p,i}$, meaning the value of $\varrho_{p,i}$ for which Margrabe's formula reproduces the stack model price.  As Figures \ref{ComparisonH} and \ref{ComparisonT} suggest, for high (positive) values of $\varrho$ in the stack, it may be impossible for Margrabe to reproduce the price, for any $\varrho_{p,i}\in[-1,1]$.  In such cases, implied correlation does not exist.  However, $\varrho^{imp}_{p,i}$ typically exists for most values of $\varrho$, and can be understood as a convenient way of measuring (or quoting) the gap between Margrabe and the stack model price. 

In Figures \ref{ImpCorrXiG}-\ref{ImpCorrMuD}, we investigate implied correlation `smiles' (against $h_c$) for a dark spread option in Scenario I and with $\varrho=0$.  In Figure \ref{ImpCorrXiG} we first vary the relative capacities of coal and gas (with $\bar{\xi}=1$ throughout). In all cases Margrabe overprices the spread since $\varrho^{imp}_{p,c}>0$, but the difference is much larger when coal is the dominant technology. As we approach the case of a dark spread in a fully gas driven market ($\bar{\xi}^g$ near 1), Margrabe approaches the stack price ($\rho^{imp}$ near 0).  In Figure \ref{ImpCorrMuD} we assume $\bar{\xi}^c=\bar{\xi}^g=0.5$, but instead vary $\mu_d$ (with $\sigma_d$ now 0.12).  We see that the implied correlation has a slight downward (upward) skew if demand is high (low), and a fairly symmetric `frown' for $\mu_d=0.5$.  Figures \ref{ImpCorrSurf1}-\ref{ImpCorrSurf2} plot implied correlation as a function of both $h$ and $T$ for Scenarios I and II.  When given fuel forward curves as inputs (Scenario II), we can observe a distinctive tilt in the implied correlation surface for long maturities.

\subsection{Power Plant Valuation}\label{str:power_plants}

 We conclude this analysis with an investigation into the bid stack model's predictions for power plant valuation under our chosen scenarios. A generating unit of fuel type $i \in I$, with heat rate $h_i$, can be approximated as a sum of spread options on spot power (cf. \cite{aEydeland2003})).  Letting $\{T_j:j\in J\}$ represent all future hours of production over the plant's life, the plant value $(VP_t)$ for $t\in[0,T]$ is then
\begin{equation*}
VP_t = \sum_{j\in J}\exp\left(-r (T_j-t)\right)\E^{\mathbb{Q}}\left[\left.\left(P_{T_j} - h_i S^i_{T_j}\right)^+\right|\mathcal{F}_t\right]
\end{equation*} 
 While this approximation technique ignores complicated operational constraints, it is consistent with our approximation of the electricity price setting mechanism itself, since a plant bidding at cost every day receives exactly this string of payoffs in our model. 
 
\begin{figure}[htbp]
     \centering
  \subfloat[Scenario I]{\label{PowerPlant1}\includegraphics[width=0.5\textwidth]{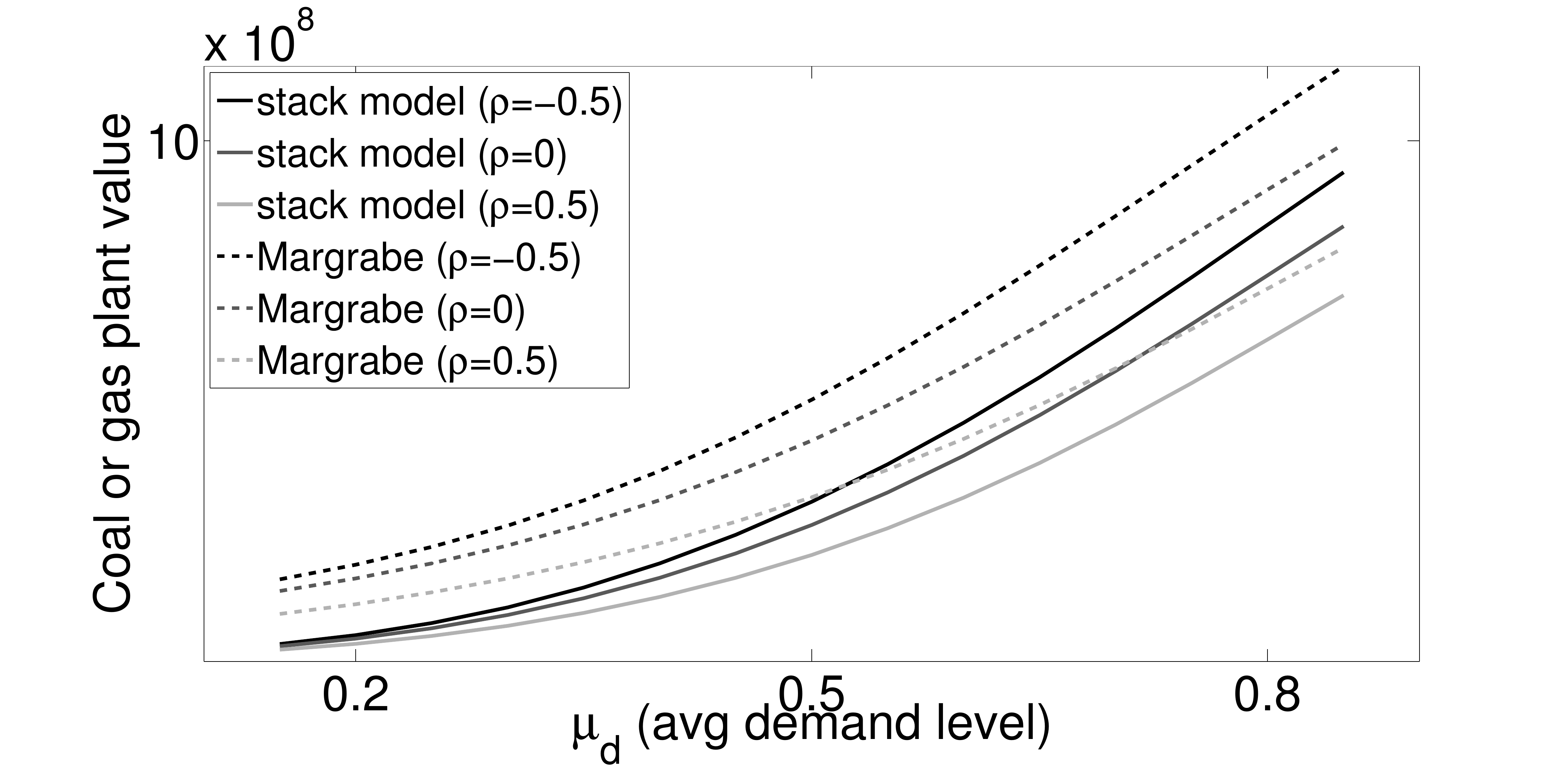}}
  \subfloat[Scenario II (Coal Plant)]{\label{PowerPlant2}\includegraphics[width=0.5\textwidth]{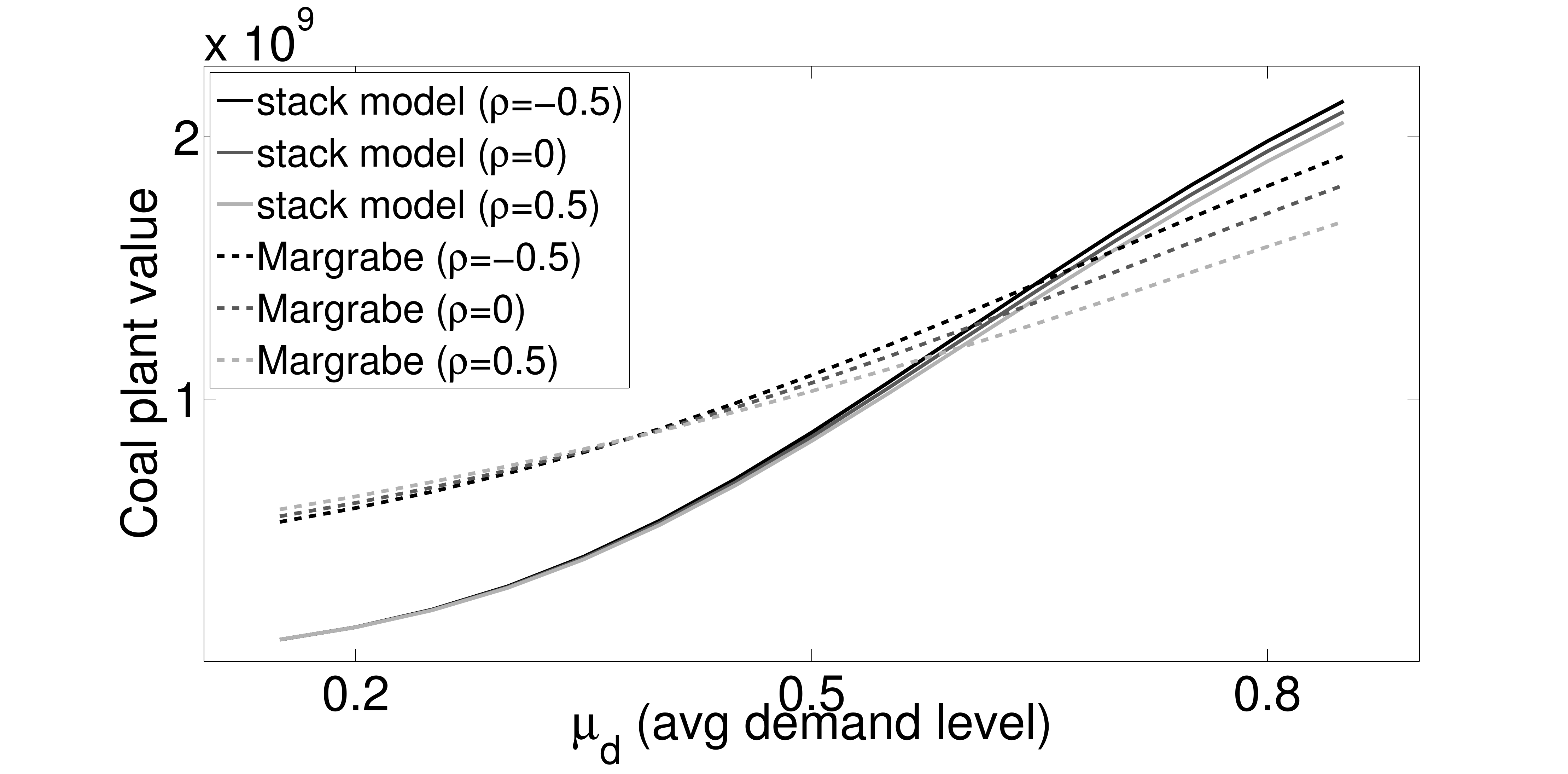}}\\
  \subfloat[Scenario III (Coal Plant)]{\label{PowerPlant3}\includegraphics[width=0.5\textwidth]{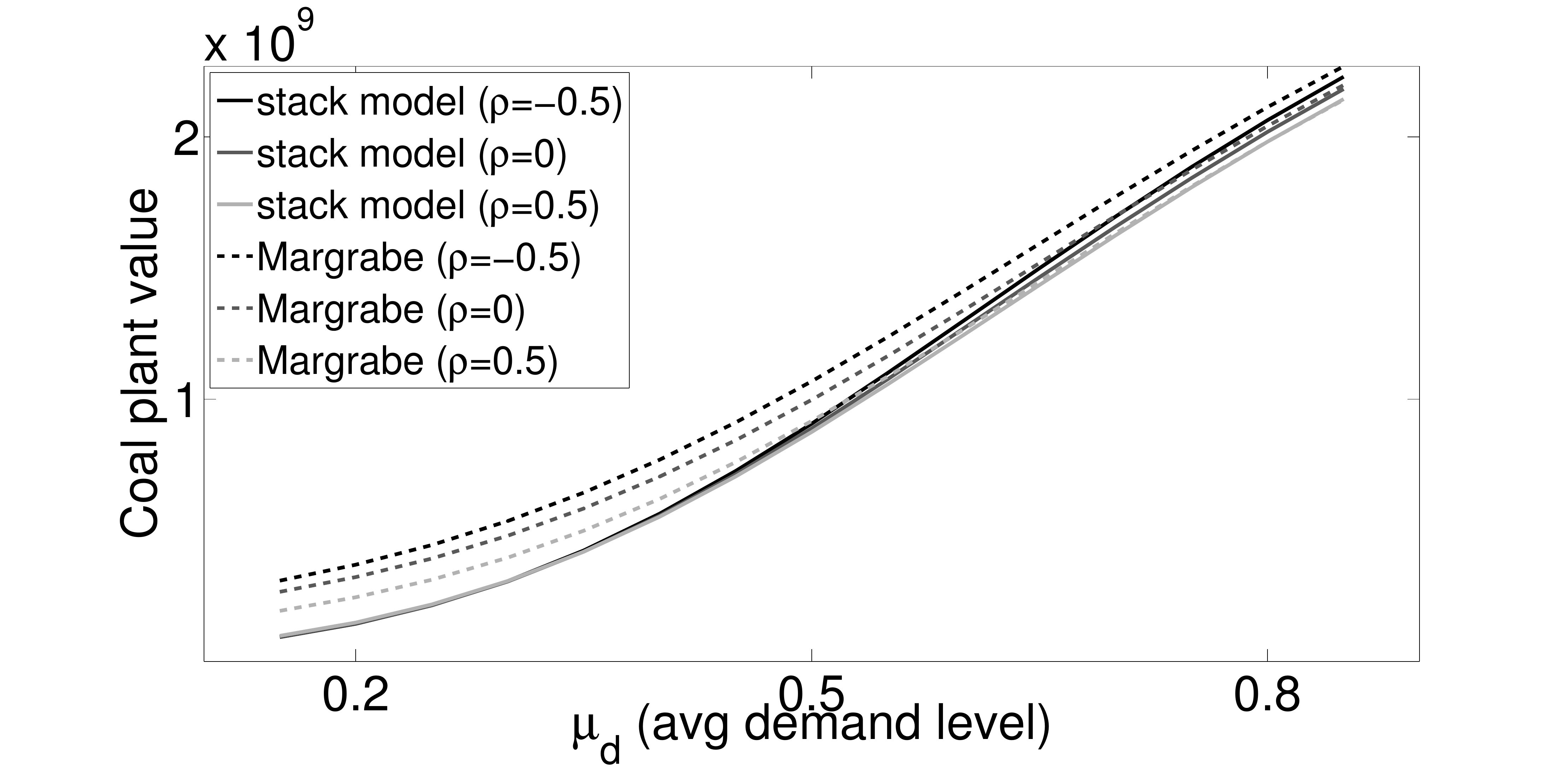}}
  \subfloat[Scenario III (Gas Plant)]{\label{PowerPlant4}\includegraphics[width=0.5\textwidth]{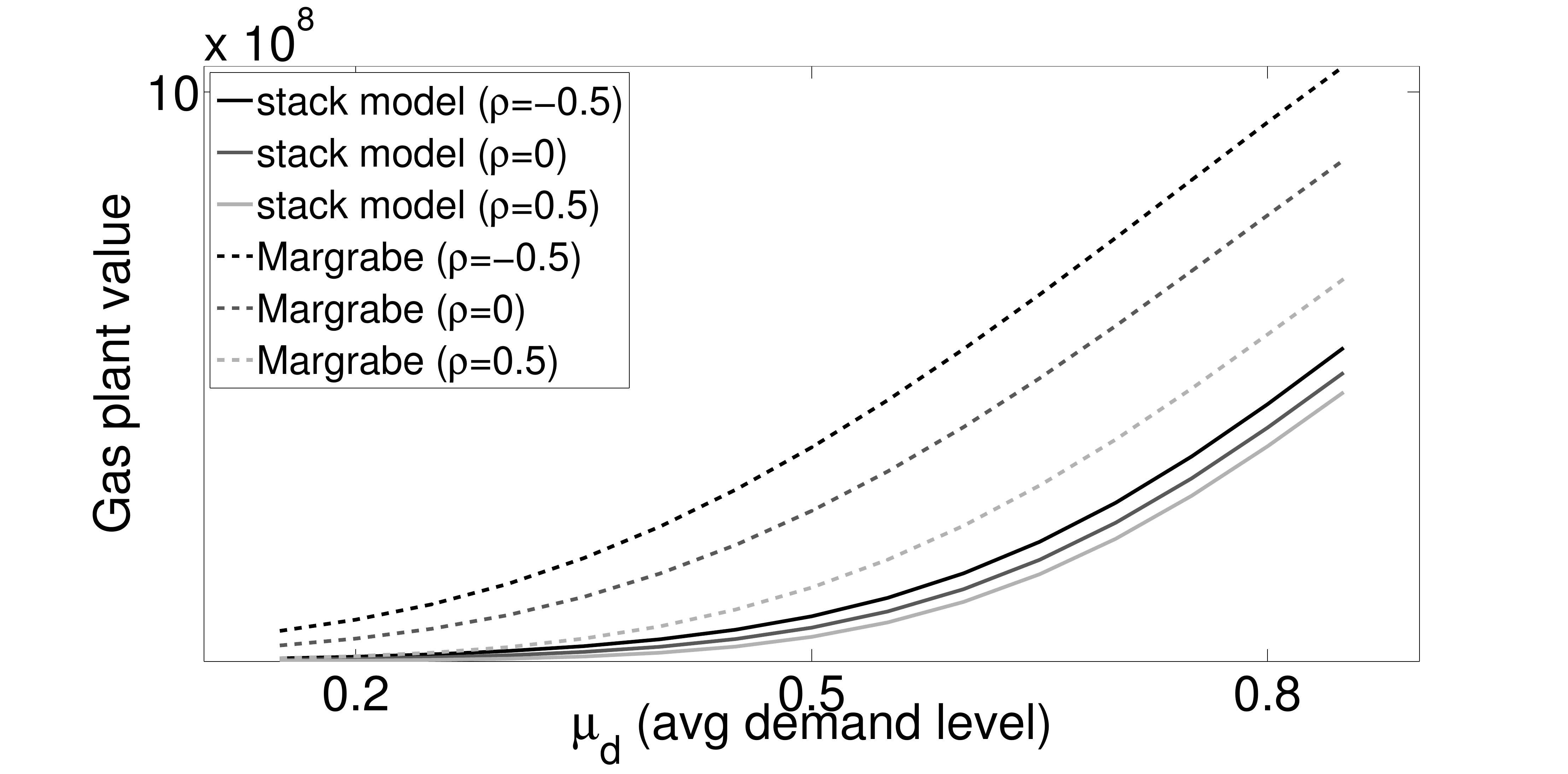}}
   \caption{Power Plant Value (3yrs, 1000MW) as a function of $\mu_d$ for various scenarios}
\label{PowerPlantPlots}
\end{figure}
In Figure \ref{PowerPlantPlots} we value a 1000 MW power plant with a life of three years in the Margrabe case and stack model, exploiting our closed-form formulas for reasonable computation time.  Instead of making an arbitrary assumption about the important periodicities of power demand (which vary from market to market) we use a fixed mean $\mu_d$ for all hours in each calculation, but investigate the resulting plant value as a function of $\mu_d$.  We also fix $h_i=\exp(k_i+m_i\bar{\xi}^i/2)$ throughout.  In Scenario I (Figure \ref{PowerPlant1}), Margrabe prices consistently higher than the stack model as expected.  Figure \ref{PowerPlant2} considers Scenario II where fuel forward prices pull future gas and coal bids in opposite directions, \emph{and} Margrabe matches the distribution of $P_T$ only from history.  Remarkably, for long enough maturities and high enough $\mu_d$, Margrabe sometimes underprices a coal power plant relative to the stack model.  Here the stack captures that gas is likely to be the marginal fuel in the future, with coal plants operating near full capacity.

Another interesting case to consider is Scenario III, in which there is very little overlap between coal and gas bids to begin with, and coal is likely to remain below gas in the merit order. Figures \ref{PowerPlant3}-\ref{PowerPlant4} reveal the result of this change.  Unsurprisingly, for a gas power plant (Figure \ref{PowerPlant4}) the deviation between the stack model and Margrabe is large, since the gas plant has little chance of being called upon to produce power.  On the other hand, the difference between Margrabe and the stack model is much less for the coal plant, and the stack model price appears to converge to Margrabe for high demand (similarly to Figure \ref{ImpCorrXiG} for high $\bar{\xi}^g$).  The reason for this is that when coal is always below gas in the stack and demand always high, then the power price can be approximated by the gas stack alone.  Hence, power price should be close to lognormal and the correlation between power and coal close to that of gas and coal.  Furthermore, as there is no mismatch of mean or variance (as there was in \ref{PowerPlant2}), under such a scenario, Margrabe's formula should give a very similar price to the stack model.

Finally, it is important to remember the heavy-tailed nature of most electricity spot prices, which makes a lognormal distribution for $P_T$ highly questionable.  As discussed in Section \ref{str:spikes}, the bid stack model allows for a straightforward extension to capture spikes (or negative prices) consistently without limiting the availability of derivative pricing formulae.  To illustrate the impact of spikes, Figure \ref{PowerPlantPlotsSpikes} plots the same scenarios as in Figures \ref{PowerPlant1} and \ref{PowerPlant4}, except using our extended stack model\footnote{This implies a spike of about \$150 in the event $X_t=1.1\,\bar{\xi}$, and so is a fairly conservative choice.  Note also that we still match the mean of the Margrabe formula (i.e., calibrate to the stack model's electricity forward curve), but only match the variance to the case of no spikes, since otherwise the fit to the new variance (with spikes) becomes infeasible for large $\mu_d$.} with parameter $m_s=50$ (and $\sigma_d$ reduced from 0.2 to 0.1).  The impact of spikes is to shift the stack model valuation closer to Margrabe, but only in the case of high $\mu_d$.  This is intuitive of course because spikes are only present for high enough demand.  Since the stack model produces a heavier-tailed distribution for $P_T$ than the lognormal, very high $\mu_d$ can lead to stack model valuation even above Margrabe.  

\begin{figure}[htbp]
  \centering
  \subfloat[Scenario I]{\label{PowerPlant1spikes}\includegraphics[width=0.5\textwidth]{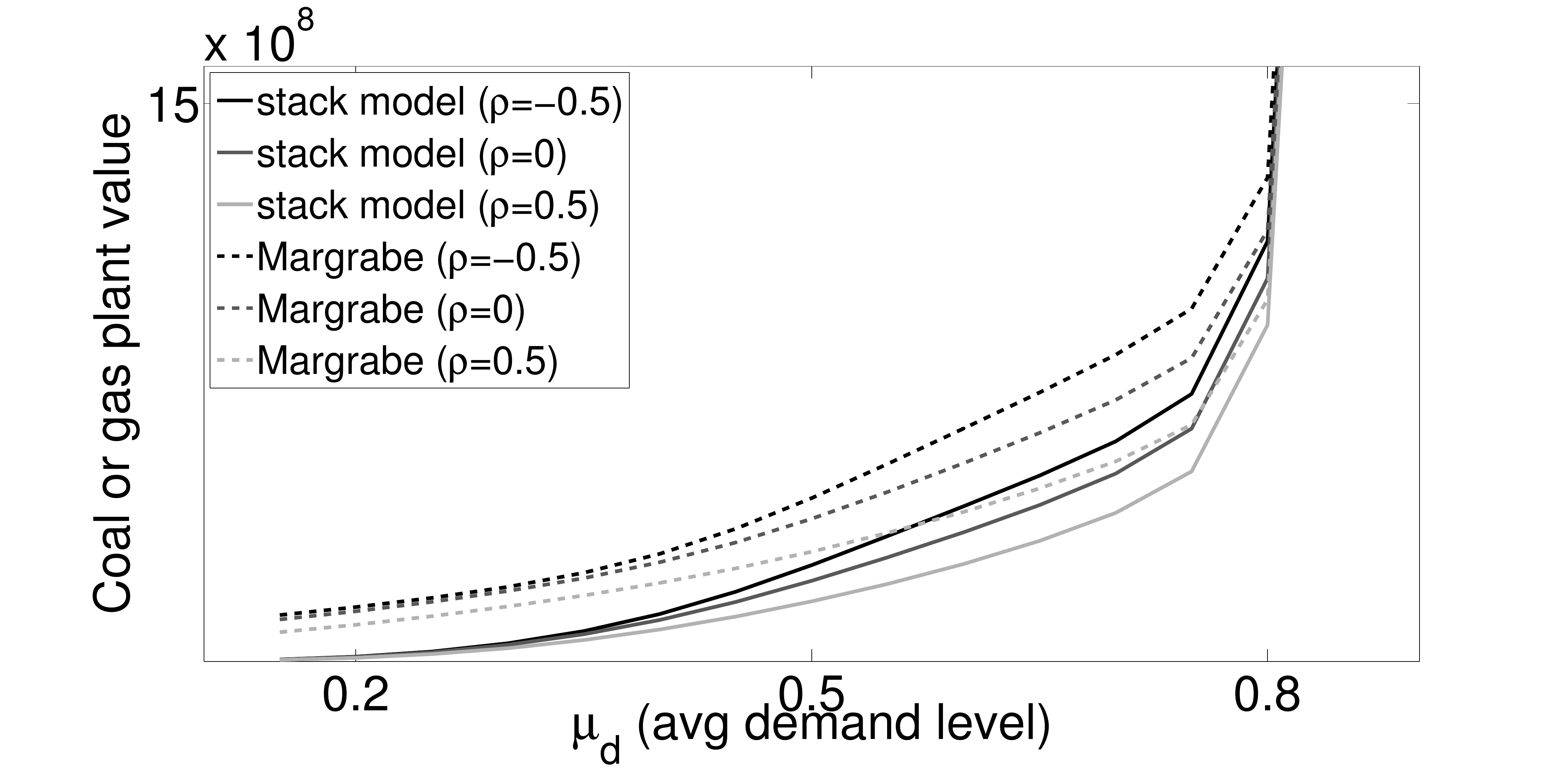}}
  \subfloat[Scenario III (Gas Plant)]{\label{PowerPlant2bspikesB}\includegraphics[width=0.5\textwidth]{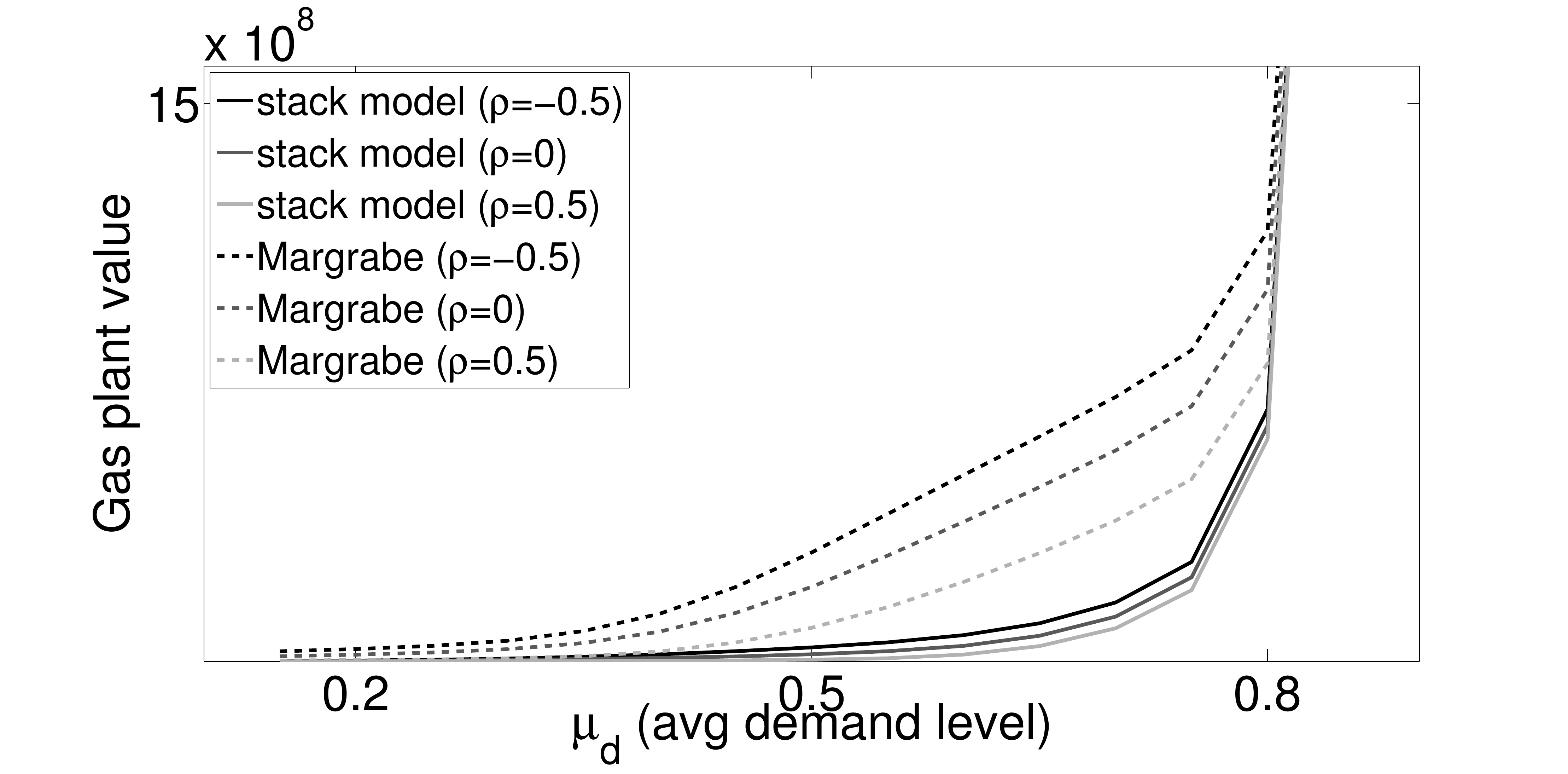}}
   \caption{Power plant value with stack model extended to include spike regime}
\label{PowerPlantPlotsSpikes}
\end{figure}

\section{Conclusion}\label{str:conclusion}

The valuation and hedging of both physical and financial assets in electricity markets is a complex and highly state-dependent challenge, particularly for medium to long term problems in markets driven by multiple underlying fuels.  As we have seen in the examples of \textsection \ref{str:pp_valuation}, it is important to be aware of the current merit order and resulting bid stack shape, as well as likely or possible changes to this order in the future.  Observed fuel forward curves can provide valuable information for this purpose, but cannot be incorporated easily into traditional reduced-form models for power prices.  On the other hand, a structural approach maintains a close link with the physical characteristics of the electricity market, allowing for the inclusion of a variety of forward looking information, such as demand forecasts, or changes in the generation mix of the market, a pertinent issue in many countries nowadays.  The piecewise exponential bid stack model proposed here achieves this link, while crucially retaining closed-form expressions for forwards and spread options, as presented in Sections \ref{str:forwards} and \ref{str:spreads}.  In this way, it enjoys the benefits of a simple reduced-form model, while mimicking the complex dependence structure produced by a full production cost optimization model, for which derivative pricing is typically a computationally infeasible task.  Furthermore, the availability of explicit expressions for forwards and options is highly beneficial for many other practical purposes, such as calibration to observed market quotes, the calculation of Greeks and for risk management applications (e.g. simulating price changes for a portfolio of physical assets).  Although we focused primarily here on a simple coal and gas based market, we believe that our general framework has the flexibility to be adapted to various market conditions, as illustrated for example by our simple extension to capture extreme spikes at times of high demand, an important feature of most power markets. As we have shown through many examples, the delicate interplay between demand, capacity, and multiple fuel prices is embedded into our approach, thus providing an intuitive framework for understanding complicated interdependencies, while also helping to bridge the prevalent gap between mathematical tractability and the economics of supply and demand.  

\appendix
\section{Moments and Covariances}\label{ap:moments}
 If demand at maturity satisfies
\begin{equation*}
 D_T = \max \left(0,\min\left(\bar{\xi},X_T\right)\right),
\end{equation*}
with $X_T\sim N(\mu_d,\sigma_d^2)$ independent of $\mathbf{S}_T$, then for $t\in[0,T]$, the $n$-th moment of $P_T$ is given by
\begin{multline*}
\E[P_T^n|\mathcal{F}_t]=\Phi_1\left(\frac{-\mu_d}{\sigma_d}\right)\sum_{i\in I} b^n_i\left(0,F^i_t\right)\exp\left(\frac12(n^2-n)\sigma_i^2\right) \Phi_1\left(\frac{R^{(n)}_i(0,0)}{\sigma}\right)\\
+\Phi_1\left(\frac{\mu_d-\bar{\xi}}{\sigma_d}\right)\sum_{i\in I} b^n_i\left(\bar{\xi}^i,F^i_t\right)\exp\left(\frac12(n^2-n)\sigma_i^2\right)\Phi_1\left(\frac{-R^{(n)}_i\left(\bar{\xi}^i,\bar{\xi}^i\right)}{\sigma}\right)\\
+\sum_{i\in I} b^n_i\left(\mu_d,F_t^i\right)\exp\left(\zeta_i^{(n)}\right)\Phi_2^{2 \times 1}\left(\left[\begin{array}{c}\frac{\bar{\xi}^i-\mu_d}{\sigma_d}-n m_i\sigma_d\\\frac{-\mu_d}{\sigma_d}-n m_i\sigma_d\end{array}\right],\frac{R^{(n)}_i(\mu_d,0)-n m_i^2\sigma_d^2}{\sigma_{i,d}};\frac{m_i\sigma_d}{\sigma_{i,d}}\right)\\
+\sum_{i\in I} b^n_i\left(\mu_d-\bar{\xi}^j,F^i_t\right)  \exp\left(\zeta_i^{(n)}\right)\Phi_2^{2 \times 1}\left(\left[\begin{array}{c}\frac{\bar{\xi}-\mu_d}{\sigma_d}-n m_i\sigma_d\\\frac{\bar{\xi}^j-\mu_d}{\sigma_d}-n m_i\sigma_d\end{array}\right],\frac{-R^{(n)}_i(\mu_d-\bar{\xi}^j,\bar{\xi}^j)+n m_i^2\sigma_d^2}{\sigma_{i,d}};\frac{-m_i\sigma_d}{\sigma_{i,d}}\right)\\
+\sum_{i\in I}\delta_i b^n_{cg}\left(\mu_d,\mathbf{F}_t\right) \exp\left(\eta^{(n)}\right)\left\{-
\Phi_2^{2 \times 1}\left(\left[\begin{array}{c}\frac{\bar{\xi}^i-\mu_d}{\sigma_d}-n\gamma\sigma_d\\\frac{-\mu_d}{\sigma_d}-n\gamma\sigma_d\end{array}\right],\frac{R_i(\mu_d,0)+n\alpha_j\sigma^2-n\gamma m_i\sigma_d^2}{\delta_i\sigma_{i,d}};\frac{m_i\sigma_d}{\delta_i\sigma_{i,d}}\right)\right.\\
+\left.\Phi_2^{2 \times 1}\left(\left[\begin{array}{c}\frac{\bar{\xi}-\mu_d}{\sigma_d}-n\gamma\sigma_d\\\frac{\bar{\xi}^j-\mu_d}{\sigma_d}-n\gamma\sigma_d\end{array}\right],\frac{R_i(\mu_d-\bar{\xi}^j,\bar{\xi}^j)+n\alpha_j\sigma^2-n\gamma m_i\sigma_d^2}{\delta_i\sigma_{i,d}};\frac{m_i\sigma_d}{\delta_i\sigma_{i,d}}\right)\right\},
\end{multline*}
where $j=I\setminus \{i\}$, $\delta_i=(-1)^{\mathbb{I}_{\{i=i_+\}}}	$ as before and
\begin{align*}
R_i^{(n)}(\xi_i,\xi_j) &:= k_j+m_j\xi_j-k_i-m_i \xi_i + \log(F^j_t)-\log(F^i_t) - (n-\frac12)\sigma_i^2-\frac12\sigma_j^2+n\rho\sigma_i\sigma_j,\\
\eta^{(n)}&:=\frac{n^2}{2}\left(\gamma^2\sigma_d^2-\alpha_c\alpha_g\sigma^2\right),\qquad \text{and} \qquad
\zeta_i^{(n)}:=\frac12(n^2-n)\sigma_i^2+\frac12n^2m_i^2\sigma_d^2.
\end{align*}
A more general formula for $\E[P_T^n (S_T^c)^{n_c} (S_T^g)^{n_g}|\mathcal{F}_t]$ can be obtained similarly, allowing us to calculate for example covariances between electricity and fuels.

\section{Terms in Spread Formula}\label{ap:spread_formula}

The integrands in Proposition \ref{prop:spreadGaussianD} correspond to the dark spread price in various demand intervals, in the case of a given or known demand value.  These terms resemble those for forwards in Proposition \ref{prop:forward}, but can be categorized by whether the option is always, never or sometimes in the money for each case.  For the last of these three cases, we have 
\begin{align*}
 &v_{\text{low},1}\left(\xi,\mathbf{x}\right)=b_c\left(\xi,x^c\right) \Phi_1\left(R_c(\xi,0)/\sigma\right) - h_c x^c \Phi_1\left(\tilde{R}_c\left(h(\xi)\right)/\sigma\right)\\
&\qquad+b_{cg}\left(\xi,\mathbf{x}\right)\exp\left(\sigma^2_{\alpha_c\alpha_g}\right) \left[1-\Phi_1\left(R_c(\xi,0)/\sigma-\alpha_g\sigma\right)- \Phi_1\left(\tilde{R}_g\left(h(\xi)\right)/\sigma-\alpha_c\sigma\right)\right],\\
 &v_{\text{high},2}\left(\xi,\mathbf{x}\right)= b_g(\xi-\bar{\xi}^c,x^g)\Phi_1\left(-R_g\left(\xi-\bar{\xi}^c,\bar{\xi}^c\right)/\sigma\right)- h_c x^c \Phi_1\left(\tilde{R}_c\left(h(\xi)\right)/\sigma\right)\\
 &\quad+b_{cg}\left(\xi,\mathbf{x}\right)\exp\left(-\frac{\alpha_c\alpha_g\sigma^2}{2}\right)\left[\Phi_1\left(R_g\left(\xi-\bar{\xi}^c,\bar{\xi}^c\right)/\sigma-\alpha_c\sigma\right)- \Phi_1\left(\tilde{R}_g\left(h(\xi)\right)/\sigma-\alpha_c\sigma\right)\right],
 \end{align*}
and $v_{\text{mid},2,c}=v_{\text{low},1}$, $v_{\text{mid},2,g}=v_{\text{high},2}$, where $ h(\xi)=\left(\log h_c-\beta-\gamma \xi\right)/\alpha_g$.  Then 
 \begin{equation*}
 v_{\text{mid},1}\left(\xi,\mathbf{x}\right)=f_{\text{mid}}\left(\xi,\mathbf{x}\right)-h_c x^c, \qquad  v_{\text{high},1}\left(\xi,\mathbf{x}\right)=f_{\text{high}}\left(\xi,\mathbf{x}\right)-h_c x^c 
 \end{equation*}
correspond to cases that are always in the money, while the following are never in the money:
 \begin{equation*}
 v_{\text{low},2}\left(\xi,\mathbf{x}\right)=v_{\text{mid},3}\left(\xi,\mathbf{x}\right)=0
 \end{equation*}

\section*{Acknowledgements}
This work was partially supported by the National Science Foundation grant DMS-0739195.

\addcontentsline{toc}{chapter}{Bibliography}

\bibliography{ReferencesMAFE}
\bibliographystyle{spmpsci}      

\end{document}